\documentclass[sigconf]{acmart}
\AtBeginDocument{%
  \providecommand\BibTeX{{%
    \normalfont B\kern-0.5em{\scshape i\kern-0.25em b}\kern-0.8em\TeX}}}

\setcopyright{acmcopyright}
\copyrightyear{2018}
\acmYear{2018}
\acmDOI{XXXXXXX.XXXXXXX}

\acmPrice{15.00}
\acmISBN{978-1-4503-XXXX-X/18/06}

\usepackage[T1]{fontenc}
\usepackage{graphicx}
\usepackage{amsmath}
\usepackage{dcolumn}
\usepackage{bm, color}
\usepackage{graphicx}
\usepackage{diagbox} 
\usepackage{caption}
\usepackage{subcaption}
\usepackage{algpseudocode}
\usepackage{algorithmicx}
\usepackage{algorithm}
\usepackage{makecell}
\usepackage{multirow}
\usepackage{hyperref}
\usepackage{smartdiagram}
\usepackage{verbatim}
\hypersetup{
    colorlinks=true,
    linkcolor=blue,
    filecolor=blue,      
    urlcolor=blue,
    citecolor=blue,
}

\usepackage{enumitem}
\setlist[itemize]{leftmargin=5mm}
\setlist[enumerate]{leftmargin=5mm}

\usepackage{quantikz}

\usepackage{tikz}

\def\COMPILETIKZ{0}

\usetikzlibrary{external}

\ifthenelse{\COMPILETIKZ=1}{
\tikzexternalize[prefix=tikz/]
\tikzset{external/system call={pdflatex \tikzexternalcheckshellescape -halt-on-error -interaction=batchmode -jobname "\image" "\texsource"}}}

\newcommand{\tikzinput}[2]{
\ifthenelse{\COMPILETIKZ=1}
{\tikzsetnextfilename{#1}\input{#2}}
{\includegraphics{./tikz/#1.pdf}}
}

\usetikzlibrary{positioning,patterns}
\usetikzlibrary{decorations.pathreplacing,decorations.pathmorphing}
\colorlet{colred}{red!22}
\colorlet{colcyan}{cyan!22}
\colorlet{colV}{blue!40}
\colorlet{colBorder}{gray!70}
\tikzset
  {mybox/.style=
    {rectangle,rounded corners,drop shadow,minimum height=1cm,
     minimum width=2cm,align=center,fill=#1,draw=colBorder,line width=1pt
    },
   myarrow/.style=
    {draw=#1,line width=3pt,-stealth,rounded corners
    },
   mylabel/.style={text=#1}
}

\usepackage{pgfplots}
\pgfplotsset{
    compat=1.5,
    grid=major,
    x label style={yshift=5pt},
    title style={yshift=-8pt},
    legend style={
        font=\tiny,
        anchor=south east,
        at={(0.98,0.02)},
        opacity=0.7,
        inner xsep=1pt,inner ysep=0pt,
        nodes={inner sep=1pt,text depth=0.15em},
        draw=black!30,
        rounded corners=1pt
    },
    legend image post style={scale=0.8},
    enlarge y limits=0.1,
    enlarge x limits=0.05,
    width=0.92\linewidth,
    height=0.59\linewidth,
    cycle list={red!65,blue!65,violet!65,teal!65,black},
}

\newcommand{\True}{\textbf{true}}
\newcommand{\False}{\textbf{false}}

\newtheorem{problem}{Problem}

\renewcommand{\epsilon}{\varepsilon}

\renewcommand{\dh}{\ensuremath{\mathcal{D(H)}}}
\newcommand{\doo}{\ensuremath{\mathcal{D(O)}}}

\DeclareMathOperator{\id}{id}

\renewcommand{\u}{\ensuremath{\mathcal{U}}}
\newcommand{\e}{\ensuremath{\mathcal{E}}}
\newcommand{\f}{\ensuremath{\mathcal{F}}}

\newcommand{\h}{\ensuremath{\mathcal{H}}}
\newcommand{\s}{\ensuremath{\mathcal{S}}}

\renewcommand{\o}{\ensuremath{\mathcal{O}}}

\renewcommand{\j}{\ensuremath{\mathcal{J}}}
\renewcommand{\k}{\ensuremath{\mathcal{K}}}

\newcommand{\tr}{{\rm tr}} 
\renewcommand{\a}{\ensuremath{\mathcal{A}}}

\newcommand{\ketbra}[2]{| #1 \rangle \langle #2 |}

\copyrightyear{2023} 
\acmYear{2023} 
\setcopyright{acmlicensed}\acmConference[CCS '23]{Proceedings of the 2023 ACM SIGSAC Conference on Computer and Communications Security}{November 26--30, 2023}{Copenhagen, Denmark}
\acmBooktitle{Proceedings of the 2023 ACM SIGSAC Conference on Computer and Communications Security (CCS '23), November 26--30, 2023, Copenhagen, Denmark}
\acmPrice{15.00}
\acmDOI{10.1145/3576915.3623108}
\acmISBN{979-8-4007-0050-7/23/11}
\begin{document}

\title{Detecting Violations of Differential Privacy for Quantum Algorithms}


\author{Ji Guan}

\orcid{0000-0002-3490-0029}
\affiliation{
  \institution{State Key Laboratory of Computer Science, Institute of Software, Chinese Academy of Sciences}
  \streetaddress{}
  \city{Beijing}
  \state{}
  \country{China}
  \postcode{100190}
}
\email{guanj@ios.ac.cn}

\author{Wang Fang}
\orcid{0000-0001-7628-1185}
\affiliation{
  \institution{State Key Laboratory of Computer Science, Institute of Software, Chinese Academy of Sciences}
  \streetaddress{}
  \city{Beijing}
  \state{}
  \country{China}
  \postcode{100190}
}
\affiliation{
  \institution{University of Chinese Academy of Sciences}
  \streetaddress{}
  \city{Beijing}
  \state{}
  \country{China}
  \postcode{100049}
}
\email{fangw@ios.ac.cn}

\author{Mingyu Huang}
\orcid{0009-0000-3219-1380}
\affiliation{
  \institution{State Key Laboratory of Computer Science, Institute of Software, Chinese Academy of Sciences}
  \streetaddress{}
  \city{Beijing}
  \state{}
  \country{China}
  \postcode{100190}
}
\affiliation{
  \institution{University of Chinese Academy of Sciences}
  \streetaddress{}
  \city{Beijing}
  \state{}
  \country{China}
  \postcode{100049}
}
\email{huangmy@ios.ac.cn}

\author{Mingsheng Ying}
\orcid{0000-0003-4847-702X}
\affiliation{
  \institution{State Key Laboratory of Computer Science, Institute of Software, Chinese Academy of Sciences}
  \streetaddress{}
  \city{Beijing}
  \state{}
  \country{China}
  \postcode{100190}
}
\affiliation{
  \institution{Department of Computer Science and Technology, Tsinghua University}
  \streetaddress{}
  \city{Beijing}
  \state{}
  \country{China}
  \postcode{100084}
}
\email{yingms@ios.ac.cn}

\renewcommand{\shortauthors}{J. Guan et al.}
\begin{abstract}
Quantum algorithms for solving a wide range of practical problems have been proposed in the last ten years, such as data search and analysis, product recommendation, and credit scoring. The concern about privacy and other ethical issues in quantum computing naturally rises up. In this paper, we define a formal framework for 
detecting violations of differential privacy for quantum algorithms. A detection algorithm is developed to verify whether a (noisy) quantum algorithm is differentially private and automatically generates bugging information when the violation of differential privacy is reported. The information consists of a pair of quantum states that violate the privacy, to illustrate the cause of the violation. Our algorithm is equipped with Tensor Networks, a highly efficient data structure, and executed both on TensorFlow Quantum and TorchQuantum which are the quantum extensions of famous machine learning platforms --- TensorFlow and PyTorch, respectively. The effectiveness and efficiency of our algorithm are confirmed by the experimental results of almost all types of quantum algorithms already implemented on realistic quantum computers, including quantum supremacy algorithms (beyond the  capability of classical algorithms), quantum machine learning models, quantum approximate optimization algorithms, and variational quantum eigensolvers  with up to 21 quantum bits. 
\end{abstract}



\maketitle
\pagestyle{empty}
\thispagestyle{empty}
\section{Introduction}
\textbf{Quantum Algorithms and Quantum Machine Learning Models}: A quantum algorithm is an algorithm that runs on a realistic model of quantum computation. The most commonly used quantum computational model is the quantum circuit model. A series of quantum algorithms have been proposed to speed up the classical counterparts to solve fundamental problems, such as Grover's algorithm~\cite{grover1996fast} for searching an unstructured database, Shor's algorithm~\cite{shor1994algorithms} for finding the prime factors of an integer and HHL (Aram Harrow, Avinatan Hassidim, and Seth Lloyd) algorithm~\cite{harrow2009quantum} for solving systems of linear equations. In recent years, motivated by the huge success of classical machine learning models in practical applications, quantum machine learning models (also known as well-trained quantum machine learning algorithms) are proposed to accelerate solving the same classical tasks. Specifically, like the classical models, a bulk of corresponding quantum learning models have been defined and trained on existing quantum hardware or simulators of quantum computation on classical supercomputers. Examples include  quantum support vector machines~\cite{Rebentrost2014},
quantum convolution neural networks~\cite{cong2019quantum}, quantum recurrent neural networks~\cite{bausch2020recurrent}, quantum generative adversarial networks~\cite{dallaire2018quantum} and quantum reinforcement learning networks~\cite{dong2008quantum}. Subsequently, these models have been tested to solve a wide range of real-world problems, such as fraud detection (in transaction monitoring)~\cite{liu2018quantum,di2021quantum}, credit assessments (risk scoring for customers)~\cite{unknown,milne2017optimal} and handwritten digit recognition~\cite{broughton2020tensorflow}. On the other hand, a series of quantum machine learning algorithms without the classical counterparts have also been designed to solve specific problems. For example,  quantum approximate optimization algorithm  (QAOA) is a toy model of quantum annealing and is used to solve problems in graph theory~\cite{farhi2014quantum}, variational quantum eigensolver (VQE) applies classical optimization to minimize the energy expectation of an ansatz state to find the ground state energy of a molecule~\cite{peruzzo2014variational}. Furthermore, based on the famous classical machine learning training platforms --- TensorFlow and Pytorch, two quantum training platforms have been established:  TensorFlow Quantum~\cite{broughton2020tensorflow} and TorchQuantum~\cite{wang2022quantumnas}, respectively. 

The rapid development of quantum hardware  enables those more and more experimental implementations of the algorithms mentioned above on concrete problems have been achieved~\cite{google2020hartree,harrigan2021quantum}. Notably, quantum supremacy (or advantage beyond classical computation) was proved by Google's quantum computer \textit{Sycamore} with 53 noisy superconducting qubits (quantum bits) that can do a  sampling task in 200 seconds, while the same task would cost (arguably) 10,000 years on the largest classical computer~\cite{arute2019quantum}.  
A type of Boson sampling was performed on USTC's quantum computer \textit{Jiuzhang} with 76 noisy photonic qubits 
in 20 seconds that would take 600 million years for a classical computer~\cite{zhong2020quantum}. These experiments demonstrate the power of quantum computers with tens to hundreds of qubits in the current \emph{Noisy Intermediate-Scale Quantum (NISQ)} era where quantum noises  cannot be avoided. Meanwhile, more and more quantum cloud computing platforms (e.g. IBM's Qiskit Runtime and Microsoft's Azure Quantum) are available for public use to implement quantum algorithms on realistic quantum chips.  

\textbf{Differential Privacy: From Classical to Quantum}: Differential privacy has become a de facto standard evaluating an algorithm for protecting the privacy of individuals. It ensures that any individual’s information has very little influence on the output of the algorithm. Based on this intuition, the algorithmic  foundation of differential privacy in classical (machine learning) algorithms has been  established~\cite{dwork2014algorithmic,ji2014differential}. However,  developing algorithms with differentially private guarantees is very subtle and error-prone. Indeed, a large number of published algorithms violate differential privacy. This situation boosts the requirement of a formal framework for verifying the differential privacy of classical algorithms. Various verification techniques have been extended into this context
~\cite{barthe2013verified,barthe2016advanced,barthe2014proving,barthe2016proving,barthe2012probabilistic,barthe2013beyond}. Furthermore,   a counterexample generator for the failure in the verification can be  provided for the debugging purpose~\cite{ding2018detecting}. 


With more and more  applications, the privacy issue of quantum algorithms also rises. Indeed, from the viewpoint of applications, this issue is even more serious than its classical counterpart since it is usually hard for the end users to understand quantum algorithms.   
Inspired by its great success in applications,  the notion of differential privacy has recently been extended to quantum computation, and some fundamental algorithmic results for computing privacy parameters have been obtained~\cite{zhou2017differential,aaronson2019gentle,hirche2022quantum} in terms of different definitions of the similarity between quantum states. However, the verification and violation detecting problem of  differential privacy of quantum algorithms have not been touched in the previous works. 

{\vskip 3pt}

\textbf{Contributions of This Paper}: In this work, we define a formal framework for the verification of differential privacy for quantum algorithms in a principled way. Specifically, our main contributions are as follows: 

\begin{enumerate}\item[(1)] \textit{\textbf{Algorithm}}: An algorithm for  detecting violations of differential privacy for quantum algorithms is developed. More specifically, this algorithm can not only  efficiently check whether or not a (noisy) quantum algorithm is differentially private, but also automatically generate a pair of quantum states when the violation of differential privacy is reported. These two states that break the promising differential privacy   provide us with debugging information.

\item[(2)]\textit{\textbf{Case Studies}}: Our detection algorithm is implemented both on TensorFlow Quantum~\cite{broughton2020tensorflow} and TorchQuantum~\cite{wang2022quantumnas} which are based on famous machine learning platforms --- TensorFlow and PyTorch, respectively. The effectiveness and efficiency of our algorithm are confirmed by the experimental results of almost all types of quantum algorithms already implemented on realistic quantum computers, including quantum supremacy algorithms (beyond the  capability of classical algorithms), quantum machine learning models, quantum approximate optimization algorithms, and variational quantum eigensolver algorithms  with up to 21 qubits.

\item[(3)]\textit{\textbf{Byproducts}}: We show that quantum noises can be used to protect the privacy of quantum algorithms as in  the case of classical algorithms, and establish a composition theorem of quantum differential privacy for handling larger quantum algorithms in a
modular way. 
 \end{enumerate}


\subsection{Related Works and Challenges}\label{sec:related_works}

{\vskip 3pt}
{\bf Detecting Violations for Classical Algorithms:} 
 Detecting the violations of differential privacy for classical (randomized) algorithms has been studied in~\cite{ding2018detecting}.  Their approach is to analyze the (distribution of) outputs of classical algorithms in a statistical way. Specifically, it runs a candidate algorithm many times and uses statistical tests to detect violations of differential privacy. However, such a method  has some limitations: if an algorithm satisfies differential privacy except with an extremely small probability  then it may not detect the violations. To avoid this situation appearing in the quantum world, we introduce a series of linear algebra operations to analyze the output states of quantum algorithms.  In particular, we characterize the verification of differential privacy as inequalities and solve them by computing eigenvalues and eigenvectors of some matrices,  which are indexed by a quantum measurement outcome and represent the converse (dual) implementation of quantum algorithms. As a result, our developed verification algorithm is exact (sound and complete).
 

{\vskip 3pt}

{\bf Differential Privacy for Quantum Circuits:} Quantum differential privacy was first defined in~\cite{zhou2017differential}-\cite{hirche2022quantum} for (noisy) quantum circuits. 
However, the verification and violation detection problems for quantum differential privacy were not addressed there. 

In this paper, we adapt the quantum differential privacy for quantum algorithms rather than quantum circuits, motivated mainly by our target applications. Roughly speaking, a quantum algorithm can be thought of as a quantum circuit together with a quantum measurement at the end to extract the computational outcome (classical information). Accordingly, the privacy for a circuit must be examined for all possible measurements, but the privacy for an algorithm should be defined for a fixed measurement. 
This subtle difference leads to different verification problems and solutions. 
In the case of algorithms, the verification problem can be solved by transferring the impact of 
algorithmic steps on input quantum states to
the given quantum measurement. But it seems that the same idea cannot be applied to the case of circuits because the final measurement is unknown beforehand. On the other hand, 
the counterexample generator of differential privacy constructed in this paper can be used to
 detect differential privacy violations in quantum circuits by appending certain measurements to them.

{\vskip 3pt}
\section{Preliminaries}\label{sec:preliminary}
In this section, for the convenience of the reader, we introduce basic ideas of quantum algorithms in a mathematical way. 

Roughly speaking, a quantum algorithm consists of a quantum circuit and a quantum measurement. The former is for implementing algorithmic instructions; the latter is to extract the classical information from the final state at  the end of the circuit. The computational components in the quantum algorithm can be mathematically  described by two types of matrices: (i) \emph{unitary matrices} for quantum gates and circuits; and (ii) \emph{positive semi-definite matrices} for density operators (quantum states) and (Positive Operator-Valued Measure) quantum measurements. Thus we start with a brief  introduction of these two kinds of matrices in the context of quantum computation.
\subsection{Unitary and Positive Semi-definite Matrices}
Before defining unitary and positive semi-definite matrices, we need to specify the state space we are interested in. Mathematically, a quantum algorithm works on a $2^n$-dimensional Hilbert (linear) space $\h$, where $n$ is the number of \emph{quantum bits (qubits)} (defined in the next section) involved in the algorithm. Thus,  in this paper, all linear algebra operations are based on $\h$. We choose to use  standard quantum mechanical notation instead  of that from linear algebra. This style of notation is known as the \emph{Dirac notation}, and  widely used in the field of quantum computation. For more details, we  refer to textbook~\cite{nielsen2010quantum}. 

First of all, vectors in $\h$ can be represented as the following Dirac notations:

\begin{enumerate}   
	\item $\ket{\psi}$ stands for a $2^n$-dimensional complex unit (normalized) column vector\footnote{$\ket{\psi}$ is a unit column vector if the inner product of $\ket{\psi}$ and itself is one, i.e., $\braket{\psi}{\psi}=1$} in $\h$ labelled with $\psi$;
	\item $\bra{\psi}:=\ket{\psi}^\dagger$ is the Hermitian adjoint
	(complex conjugate and transpose) of $\ket{\psi}$;
	\item $\braket{\psi_1}{\psi_2}:=(\ket{\psi_1}, \ket{\psi_2})$
	is the inner product of $\ket{\psi_1}$ and $\ket{\psi_2}$;
	\item $\ketbra{\psi_1}{\psi_2}$
 is the outer product;
	\item $\ket{\psi_1,\psi_2}:=\ket{\psi_1}\ket{\psi_2}$ is a shorthand of
    the product state $\ket{\psi_1}\otimes\ket{\psi_2}$.
\end{enumerate}

 {\bf Unitary Matrices:} In the ($2^n$-dimensional) Hilbert space $\h$, a unitary matrix $U$ is a  $2^n\times 2^n$ matrix with  $U^\dagger U=UU^\dagger=I_n$, where $U^\dagger=(U^*)^\top$ is the (entry-wise) conjugate transpose of $U$ and  $I_n$ is  the identity matrix on  $\h$.

{\bf Positive Semi-Definite Matrices:} A $2^n\times 2^n$ matrix $M$ is called \emph{positive semi-definite} if for any $\ket{\psi}\in\h$,  $\bra{\psi}M\ket{\psi}\geq 0$. Subsequently, all eigenvalues of $M$ are non-negative. That is, for any unit eigenvector $\ket{\psi}$ of $M$ ( i.e., $M\ket{\psi}=\lambda\ket{\psi}$), we have  $\lambda\geq 0$. 

Some examples of these two matrices with physical meanings will be provided in the next section for a better understanding.


\subsection{Quantum Algorithms}\label{sec_quantum_algorithms}
{\vskip 3pt}
Now we turn to review the setup of quantum algorithms in their most basic form. 
A quantum algorithm is a set of instructions solving a problem (e.g., Shor's algorithm for finding the prime factors of an integer) that can be performed on a quantum computer. Physically, the algorithm is implemented by a quantum circuit that can be executed on quantum hardware. The computational flow of the quantum algorithm is drawn in the following.

\begin{figure*}
    \centering
    \tikzinput{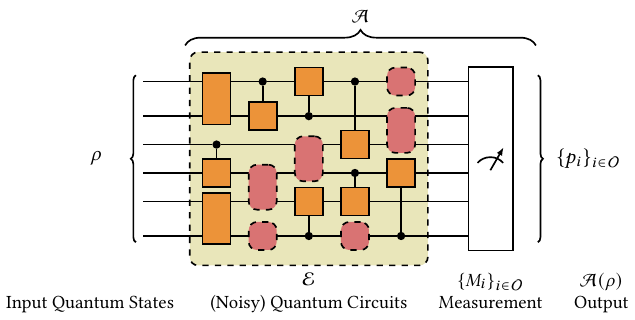}{tikz/model_qa}
    \vskip -10pt
    \caption{The Computational Model of Quantum Algorithms.}\label{fig:models}
\end{figure*}

With the notions introduced in the above subsection, we can explain the above procedures from the left side to the right one.

{\bf Input Quantum States:}  An input can be a \emph{pure quantum state}, which is mathematically modeled as a  complex unit column vector $\ket{\psi}$ in a $2^n$-dimensional Hilbert (linear) space $\h$, where $n$ denotes the number of qubits in $\ket{\psi}$. For example, a state of a qubit is a vector in a $2$-dimensional Hilbert space, written in the Dirac notation as  \begin{equation*}
    \ket{q}=\left(\begin{array}{cc}a\\ b\end{array}\right)=a\ket{0}+b\ket{1} \ {\rm with }\ \ket{0}=\left(\begin{array}{cc}1\\ 0\end{array}\right) \text{ and }\ket{1}=\left(\begin{array}{cc}0\\ 1\end{array}\right),
\end{equation*}
where complex numbers $a$ and $b$ satisfy the  normalization condition $|a|^2+|b|^2=1$. Here, the orthonormal basis $\ket{0}$, $\ket{1}$ of the Hilbert space  corresponds to the digital value $0$, $1$ of a bit  in classical computers, respectively.

On a NISQ hardware, noises are unavoidable, and a pure state $\ket{\psi}$ on $\h$ may collapse into a \emph{mixed state}, represented as an \emph{ensemble} $\{(p_{k},\ket{\psi_{k}})\}_{k}$, meaning that it is in  $\ket{\psi_{k}}$ with probability $p_{k}$. Mathematically, the ensemble can be described by a $2^n\times 2^n$ positive semi-definite matrix:
\[\rho = \sum_{k} p_{k}\ketbra{\psi_{k}}{\psi_{k}}\]
 with unit trace in the $2^n$-dimensional Hilbert (linear) space $\h$, i.e.,  $\tr(\rho) = 1$, where trace $\tr(\rho)$ of $\rho$ is defined as the summation of diagonal elements of $\rho$. 
 We use $\dh$ to denote the set of all (mixed) quantum states in $\h$.

{\vskip 3pt}

{\bf (Noisy) Quantum Circuits:} 
The computational part (without the final measurement) of a quantum algorithm can be described by a quantum circuit.  A quantum circuit $U$ consists of a sequence (product) of \emph{quantum logic gates} $U_i$, i.e., $U= U_d\cdots U_1$ ( See the orange boxes of the quantum circuit in Fig.~\ref{fig:models}). Here $d$ is the depth of the circuit $U$, and each $U_i$ is mathematically modeled by a unitary matrix. For an input $n$-qubit state $\rho$, the output of the circuit is a quantum state of the same size: \begin{equation}\label{eq:noiseless_circuit}
    \rho'=U\rho U^\dagger.
\end{equation}
\begin{example}
A set of typical quantum logic gates used in this paper are listed in the following.
\begin{itemize}
    \item [(I)] 1-qubit (parameterized) logic gates ($2\times 2$ unitary matrices):
    \begin{gather*}
        \begin{aligned}
        X &= \begin{pmatrix}
            0 & 1 \\ 1 & 0
        \end{pmatrix} & 
        Y &= \begin{pmatrix}
            0 & -i \\ i & 0
        \end{pmatrix} &
        Z &= \begin{pmatrix}
            1 & 0 \\ 0 & -1
        \end{pmatrix} \\
        H &= \frac{1}{\sqrt{2}}\begin{pmatrix}
            1 & 1 \\ 1 & -1
        \end{pmatrix} &
        S &= \begin{pmatrix}
            1 & 0 \\ 0 & i
        \end{pmatrix} &
        T &= \begin{pmatrix}
            1 & 0 \\ 0 & e^{i\pi/4}
        \end{pmatrix}.
    \end{aligned}
    \end{gather*}
      \item [(II)] 1-qubit  rotation gates that are  rotation operators along $x,y,z$-axis by angle $\theta$, respectively:
    \begin{gather*}
    \begin{aligned}
        R_x(\theta) &= e^{-i\theta X/2} = \cos\frac{\theta}{2}I-i\sin\frac{\theta}{2}X = \begin{pmatrix}
            \cos\frac{\theta}{2} & -i\sin\frac{\theta}{2} \\
            -i\sin\frac{\theta}{2} & \cos\frac{\theta}{2}
        \end{pmatrix} \\
        R_y(\theta) &= e^{-i\theta Y/2} = \cos\frac{\theta}{2}I-i\sin\frac{\theta}{2}Y = \begin{pmatrix}
            \cos\frac{\theta}{2} & -\sin\frac{\theta}{2} \\
            \sin\frac{\theta}{2} & \cos\frac{\theta}{2}
        \end{pmatrix} \\
        R_z(\theta) &= e^{-i\theta Z/2} = \cos\frac{\theta}{2}I-i\sin\frac{\theta}{2}Z = \begin{pmatrix}
            e^{-i\theta/2} & 0 \\ 0 & e^{i\theta/2}
        \end{pmatrix}.
    \end{aligned}
    \end{gather*}
    Rotation gates $R_x(\theta),R_y(\theta),R_z(\theta)$ are widely used to encode classical data into quantum states and also construct quantum machine learning models (parameterized quantum circuits). These will be detailed in the later discussion.

\item [(III)] 2-qubit Controlled-U gates ($4\times 4$ unitary matrices): For any 1-qubit logic gate $U$, we can get a 2-qubit logic gate --- controlled-$U$ (CU) gate, applying $U$ on the second qubit (the target qubit) if and only if the first qubit (the control qubit) is $\ket{1}$. See the following instances:
\begin{enumerate}
    \item CNOT: CX gate is also known as controlled NOT (CNOT) gate and has a special circuit representation:
    \[
    \text{CX}
    =\begin{quantikz}[row sep=2.2mm, column sep=2.3mm,]
    & \ctrl{1} & \qw\\
      &\targ{} & \qw
    \end{quantikz}
    =\begin{pmatrix}1&0&0&0\\
    0&1&0&0\\
    0&0&0&1\\
    0&0&1&0
    \end{pmatrix}.
    \]

    \item CZ gate:
    \[
    \text{CZ}
    =\begin{quantikz}[row sep=2.2mm, column sep=2.3mm,]
    & \ctrl{1} & \qw\\
      &\control\qw & \qw
    \end{quantikz}
    =\begin{pmatrix}1&0&0&0\\
    0&1&0&0\\
    0&0&1&0\\
    0&0&0&-1
    \end{pmatrix}
    \]
    \item Controlled parameterized gates: For example, the controlled Pauli X rotation gate with rotation angle $\theta$ is:
    \[
    \begin{quantikz}[row sep=2.2mm, column sep=2.3mm,]
    & \ctrl{1} & \qw\\
      &\gate{R_x(\theta)} & \qw
    \end{quantikz}
    =\begin{pmatrix}
        1&0&0&0\\
        0&1&0&0\\
        0&0&\cos \frac{\theta}{2} & -i \sin \frac{\theta}{2} \\
        0&0&-i \sin \frac{\theta}{2} & \cos \frac{\theta}{2}
        \end{pmatrix}
    \]
    
\end{enumerate}
\end{itemize}
\end{example}

\begin{figure}[!ht]
\centering
\subcaptionbox{A simple quantum neural network to perform MNIST image classification task in TorchQuantum's tutorial.\label{examples:a}}[0.9\linewidth]{
        \scalebox{0.95}{\tikzinput{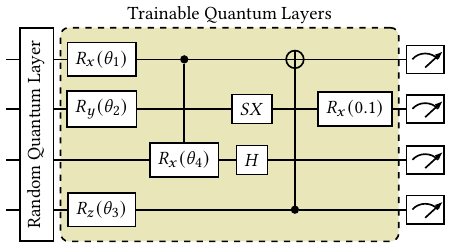}{tikz/qnn}}
}

\subcaptionbox{A quantum supremacy algorithm with a $2 \times 2$ qubits layout with four layers.\label{examples:b}}[0.9\linewidth]{
        \scalebox{0.95}{\begin{quantikz}[row sep=2.2mm, column sep=2.3mm,]
          & \gate{H} & \gate{T}\gategroup[4,steps=6,style={draw=white,dashed,
             rounded corners, inner xsep=0pt},
             background]{}     &\ctrl{1}\qw  & \qw         &\gate{T}       &\ctrl{2}\qw  &\qw            &\qw     &\meter{}\\
    & \gate{H} & \gate{T}     &\control\qw  &\ctrl{2}\qw  &\qw            & \qw         &\gate{T}       &\qw  &\meter{}\\
    & \gate{H} & \ctrl{1}\qw  &\gate{T}     &\qw          &\gate{SX} &\control\qw  &\qw            &\qw  &\meter{}\\
    & \gate{H} & \control\qw  &\gate{T}     &\control\qw  &\qw            & \qw         &\gate{T}       &\qw  &\meter{}        \end{quantikz}}
    }
    \vskip -10pt
    \caption{Examples of Quantum Machine Learning and Supremacy Algorithms}\label{examples}
\end{figure}
In quantum circuits, each quantum  gate $U_i$ only non-trivially operates on one or two qubits.
For example, if $U_i$ represents a Hadamard gate  on the first qubit, then $U_i=H\otimes I_{n-1}$, where   $I_{n-1}$ is a $2^{n-1}\times 2^{n-1}$ identity matrix applied on the rest $n-1$ qubits. See the gates in Figure~\ref{examples}.

 In the current NISQ era, a (noiseless)  quantum circuit $U$ can only have a noisy implementation modeled by a linear mapping $\e$ from $\dh$ to $\dh$ satisfying the following two conditions: 
\begin{itemize}
    \item $\e$ is trace-preserving: $\tr(\e(\rho)) = \tr(\rho)$ for all  $\rho\in\dh$; 
    \item $\e$ is completely positive: for any Hilbert space $\h'$, the trivially extended operator $\id_{\h'} \otimes \e$ maps density operators  to density operators on $\h' \otimes \h$, where  $\id_{\h'}$ is the identity map on $\h'$: $\id_{\h'}(\rho)=\rho$ for all $\rho\in\mathcal{D(H')}$.
\end{itemize}
Such a mapping $\e$ is called a \emph{super-operator} in the field of quantum computing and  admits a \emph{Kraus matrix form}~\cite{nielsen2010quantum}: there exists a finite set $\{E_k\}_{k\in\k}$ of matrices on $\h$ such that 
\[
 \e(\rho)=\sum_{k\in\k}E_k\rho E_k^\dagger \quad \textrm{ with  } \sum_{k\in\k}E_k^\dagger E_k=I_n,\]
where $\{E_k\}_{k\in\k}$ is called \emph{Kraus matrices} of $\e$. In this case, $\e$ is often represented as $\e=\{E_k\}_{k\in\k}$. Thus, for an input state $\rho$ fed into the noisy quantum circuit $\e$, the output state is:
\begin{equation}\label{eq:noisy_circuit}
    \rho'=\e(\rho).
\end{equation}
If $\e$ degenerates to a unitary matrix $U$, i.e., $\e=\{U\}$, then the above equation (evolution) is reduced to the noiseless case in Eq.~(\ref{eq:noiseless_circuit}). Briefly, we write such $\e=\{U\}$ as $\u=\{U\}$ representing noiseless quantum circuit $U$. 

Similarly to a noiseless quantum circuit $U$, a noisy quantum circuit $\e$ also consists of a sequence (mapping composition) of quantum logic (noisy) gates $\{\e_{i}\}$, i.e., $\e=\e_d\circ\cdots \circ\e_1$, where each $\e_{i}$ is either a noiseless quantum logic gate or a noisy one (e.g., the  red dashed boxes of the noisy quantum circuit in Fig.~\ref{fig:models}). See the following examples of quantum noisy logic gates in a mathematical way.
\begin{example}\label{Exa:noise}
Let us consider the following noise forming of a 1-qubit gate $U$:
\[\e_{U,p}(\rho)=(1-p)\rho+ pU\rho U^\dagger, \quad\forall \rho\in\dh\]
 where $0\leq p\leq 1$ is a probability measuring the noisy level (effect) and  $U$ is a unitary matrix. Then $\e_{U,p}$ consists of Kraus matrices $\{\sqrt{1-p}I, \sqrt{p}U\}$. Such $\e_{U,p}$ can be used to model several typical 1-qubit noises,  depending on the choice of $U$: $U=X$ for bit flip, $U=Z$ for phase flip and $U=Y=iXZ$ for bit-phase flip~\cite[Section 8.3]{nielsen2010quantum}. The depolarizing noise   combines these three noises. It is represented by 
 \[\e_{D,p}=\{\sqrt{1-p}I,\sqrt{\frac{p}{3}}X,\sqrt{\frac{p}{3}}Y,\sqrt{\frac{p}{3}}Z\},\] or equivalently
\[\e_{D,p}(\rho)=(1-p)\rho+\frac{p}{3}(X\rho X+Y\rho Y+Z\rho Z), \quad \forall \rho\in\dh.\]
\end{example}

{\bf Quantum Measurement:} At the end of each quantum algorithm, a \emph{quantum measurement} is set to extract the computational outcome (classical information). Such information is a probability distribution over the possible outcomes of the measurement. Mathematically, a quantum measurement is modeled by a set $\{M_{k}\}_{k\in\o}$ of positive semi-definite matrices on its state (Hilbert)  space $\h$ with $\sum_{k} M_{k}=I$, where $\o$ is a finite set of the measurement outcomes.  This observing process is probabilistic: if the output of the quantum circuit before the measurement is quantum  state $\rho$, then a measurement outcome $k$ is obtained with probability  
\begin{equation}\label{Eq:measure_probability}
    p_{k}=\tr(M_{k}\rho).
\end{equation}
Such measurements are known as \emph{Positive Operator-Valued Measures} and are widely used to describe the probabilities of outcomes without concerning the post-measurement quantum states (note that after the measurement, the state will be collapsed (changed), depending on the measurement outcome $k$, which is fundamentally different from the classical computation.)

By summarizing the above ideas, we obtain a general model of quantum algorithms as depicted in Fig.~\ref{fig:models}:
\begin{definition}\label{def:decision_model}
  A quantum algorithm  $\a=(\e,\{M_k\}_{k\in \o})$ is a randomized mapping $\a: \dh\rightarrow\doo$ defined by 
  \[\a(\rho)=\left\{\tr(M_k\e(\rho))\right\}_{k\in\o}\quad \forall \rho\in\dh,\]
  where:\begin{enumerate}\item $\e$ is a super-operator on Hilbert space $\h$ representing a noisy quantum circuit; 
  \item $\{M_{k}\}_{k\in \o}$ is a quantum measurement on $\h$  with $\o$ being the set of measurement outcomes (classical information);
  \item $\doo$ stands for the set of probability distributions over $\o$.\end{enumerate}

  In particular, if $\e$ represents a noiseless quantum circuit $U$ written as $\u=\{U\}$, then we call $\a=(\u,\{M_k\}_{k\in \o})$ a noiseless quantum algorithm. 
\end{definition}

According to the above definition, a quantum algorithm $\a=(\e,\{M_k\}_{k\in \o})$ is a randomized mapping, and thus we can estimate not only the distribution  $\{\tr(M_k\e(\rho))\}_{k\in\o}$ but also the summation  $\sum_{k\in\s}\{\tr(M_k\e(\rho))\}_{k\in\o}$ for any subset $\s\subseteq \o$ in a statistical way. This observation is essential in defining differential privacy for quantum algorithms in the next section. 

{\bf Quantum Encoding:} To make quantum algorithms useful for solving practical classical problems, the first step is to encode classical data into quantum states. There are multiple encoding methods, but \emph{amplitude encoding} and \emph{angle encoding} are two of the most widely used.
\begin{itemize}
\item {\bf Amplitude encoding} represents a vector \textbf{$\bar{v}$} as a quantum state \textbf{$|\bar{v}\rangle$}, using the amplitudes of the computational basis states $|i\rangle$:
$$\bar{v}=(v_1,v_2,\ldots,v_N) \rightarrow  |\bar{v}\rangle = \sum_{i=1}^N \frac{v_i}{\|\bar{v}\|} |i\rangle$$
where \textbf{$\|\bar{v}\|$} normalizes the state. This encoding uses only \textbf{$\log_2 N$ qubits} to represent an $N$-dimensional vector. However, preparing the state $|\bar{v}\rangle$ requires a deep, complex circuit beyond the current NISQ hardwares.
\item {\bf Angle encoding} encodes a vector \textbf{$\bar{v}$} by rotating each qubit by an angle corresponding to one element of \textbf{$\bar{v}$}:
$$\bar{v}=(v_1,v_2,\ldots,v_n) \rightarrow \ket{\bar{v}} = \bigotimes_{j=1}^n R(v_j)\ket{0}$$
where $R(v_j)$ rotates qubit $j$ by angle $v_j$ along some axis, i.e.,  $R$ can be one of $R_x,R_y,R_z$. This encoding uses $n$ qubits for an $n$-dimensional vector but only requires simple 1-qubit rotation gates. As an example, encoding \(\bar{v}=(\pi,\pi,\pi)\) via \(R_y\) rotations yields \(\ket{\bar{v}}=\ket{1,1,1}=\ket{1}\otimes\ket{1}\otimes\ket{1}\).
A key advantage of angle encoding is its parallelizability. Each qubit undergoes a rotation gate simultaneously, enabling encoding in constant time as shown in the following. This makes angle encoding well-suited for the current NISQ devices. Therefore, angle encoding is commonly used in the experimental implementation of quantum algorithms on existing quantum computers for solving classical computational tasks.
\end{itemize}

\[
\qquad\scalebox{0.8}{\begin{quantikz}[row sep=8pt]
    \lstick{$\ket{0}$} & \gate[style={minimum width=1.5cm}]{R(v_1)} & \qw \ R(v_1)\ket{0} &[-5pt] \rstick[4]{$\bigotimes_{j=1}^n R(v_j)\ket{0}$}\\
    \lstick{$\ket{0}$} & \gate[style={minimum width=1.5cm}]{R(v_2)} & \qw R(v_2)\ket{0} & \\
    \lstick{$\vdots$} & \vdots & \vdots & \\
    \lstick{$\ket{0}$} & \gate[style={minimum width=1.5cm}]{R(v_{n})} & \qw\ R(v_n)\ket{0} & \\
\end{quantikz}}
\]

With the above encoding methods for pure state $\ket{\bar{v}}$, we can simply obtain a mixed state to carry the classical data $\bar{v}$:
$$\rho_{\bar{v}}=\ketbra{\bar{v}}{\bar{v}}.$$

In this paper, we consider the differential privacy of quantum algorithms on NISQ computers. As such, all of our experiments in the Evaluation section (Section~\ref{Sec:evaluation}) use angle encoding to encode classical data, including credit records, public adult income dataset, and transactions dataset. 

\section{Formalizing Differential Privacy}\label{Sec:definition}
In this section, we introduce the differential privacy for quantum algorithms and clarify the relationship between it and the differential privacy for quantum circuits defined in~\cite{zhou2017differential}. For the convenience of the reader, we put all proofs of theoretical results in the appendix.

\begin{figure*}[htbp]    \includegraphics[width=.8\textwidth]{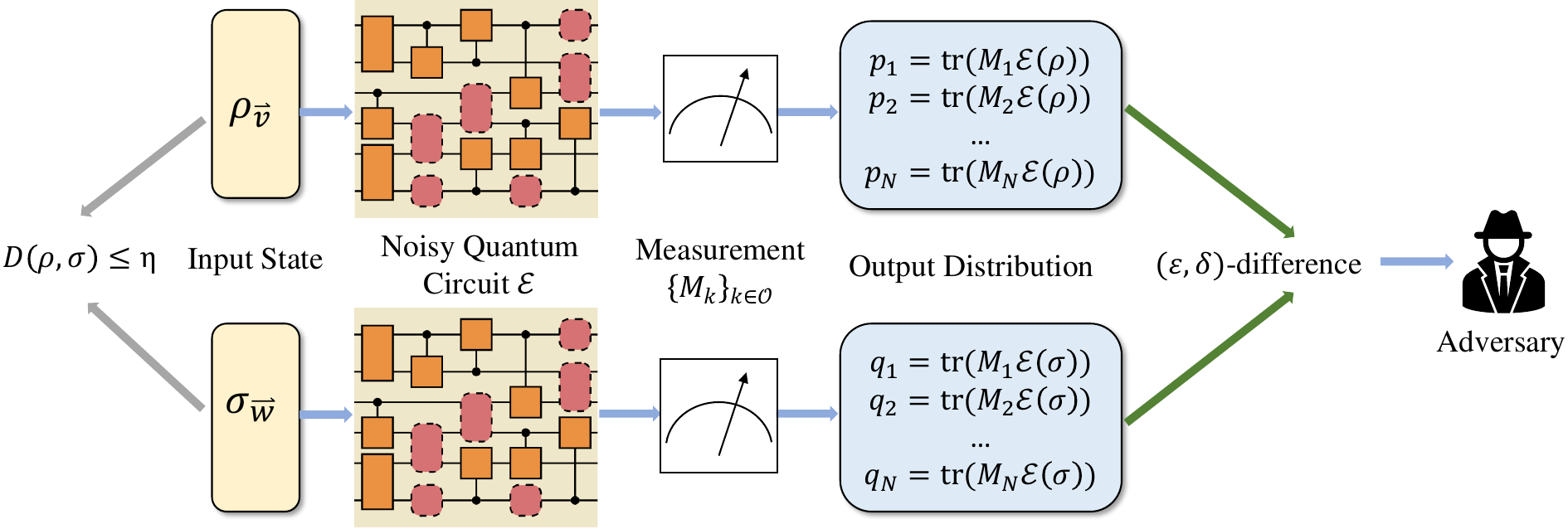}
    \caption{Quantum Differential Privacy}
    \label{pic:QDP}
\end{figure*}

Let us start by defining the differential privacy for quantum algorithms: 
\begin{definition}[Differential Privacy for Quantum Algorithms]\label{def:DP}
  Suppose we are given a quantum algorithm  $\a=(\e,\{M_k\}_{k\in\o})$ on a Hilbert space $\h$, a distance metric $D(\cdot,\cdot)$ on $\dh$, and three  small enough threshold values $ \epsilon,\delta,\eta\geq 0$. Then $\a$ is said to be  $(\epsilon,\delta)$-differentially private within $\eta$ if for any quantum states $\rho, \sigma\in\dh$ with $D(\rho,\sigma)\leq \eta$, and for any subset $\s\subseteq\o$, we have 
  \begin{equation}\label{def-DP-eq} \sum_{k\in\s}\tr(M_{k}\e(\rho))\leq \exp(\epsilon)\sum_{k\in\s}\tr(M_{k}\e(\sigma))+\delta.\end{equation}
  In particular, if $\delta=0$, we say that $\a$ is $\epsilon$-differentially private within $\eta$.
\end{definition}

 The above definition is essentially a quantum generalization of differential privacy for randomized algorithms~\cite{dwork2014algorithmic}. Thus, it shares the intuition of differential privacy discussed in~\cite{dwork2014algorithmic}: an algorithm must behave similarly on similar input states (considered as neighbors in the state space). In the quantum case, we have: 
 \begin{enumerate}
    \item $\eta$ defines the (noisy) neighboring relation between the two  input states $\rho$ and $\sigma$, i.e., $D(\rho,\sigma)\leq \eta$;
     \item $\epsilon$ and $\delta$ through Eq.(\ref{def-DP-eq}) guarantee the similarity between the outputs of $\sum_{k\in\s}\tr(M_{k}\e(\rho))$ and $\sum_{k\in\s}\tr(M_{k}\e(\sigma))$;
     \item  Since a quantum algorithm is a randomized function, it is reasonable to consider the probability $\sum_{k\in\s}\tr(M_{k}\e(\rho))$ that the output is within a subset $\s\subseteq \o$ rather than an exact value of $\tr(M_{k}\e(\rho))$. The arbitrariness of $\s\subseteq \o$  in Eq.(\ref{def-DP-eq}) ensures the differential privacy in randomized functions as the same as in the classical case~\cite{dwork2014algorithmic}. 
 \end{enumerate}

Consequently, quantum differential privacy ensures that the indistinguishabilities of any neighboring quantum states are kept by quantum algorithms. Specifically, as shown in Fig.~\ref{pic:QDP}, an adversary is hard to determine whether the input state of the algorithm was indeed $\rho$ or a neighboring state $\sigma$ such that the information revealed in the $(\epsilon,\delta)$-difference between $\rho$ and $\sigma$ in Eq.~(\ref{def-DP-eq}) cannot be easily inferred by observing the output measurement distribution of the algorithm. Furthermore, quantum encoding allows quantum states to encode classical data, so $\rho$ and $\sigma$ can be regarded as $\rho_{\bar{v}}$ and $\sigma_{\bar{w}}$ which encodes classical vectors $\bar{v}$ and $\bar{w}$. Thus the distance bound $\eta$ between $\rho_{\bar{v}}$ and $\sigma_{\bar{w}}$ can be used to represent the single element difference of classical data $\bar{v}$ and $\bar{w}$. Thus classical neighboring relation can be preserved by the quantum counterpart. Therefore, quantum differential privacy can be used as a proxy to ensure the original motivating privacy that the presence or absence of any individual data record will not significantly affect the outcome of an analysis. A concrete example is provided to detail this in the later of this section. Furthermore, this idea will be utilized in our case studies in Section~\ref{Sec:evaluation} to demonstrate how quantum noise can enhance the privacy of encoded classical data.

 It is easy to see that when considering noiseless  trivial quantum circuits (i.e., $\e = \text{id}_{\mathcal{H}}$, the identity map on $\mathcal{H}$), the above setting  degenerates to Aaronson and Rothblum's framework~\cite{aaronson2019gentle} where an elegant connection between quantum differential privacy and gentle measurements was established. In this paper, we consider a more general class of measurements, and a  connection between quantum measurements and the verification of  quantum differential privacy under quantum noise is revealed.

 By Definition~\ref{def:DP}, if a quantum algorithm $\a=(\e,\{M_k\}_{k\in\o})$ is not $(\epsilon,\delta)$-differentially private, then there exists at least one pair of quantum states $(\rho,\sigma)$ with the distance of them being within $\eta$, i.e.,  $D(\rho,\sigma)\leq \eta$, and a subset $\s\subseteq \o$ such that 
\begin{equation}
\sum_{k\in\s}\tr(M_{k}\e(\rho))> \exp(\epsilon)\sum_{k\in\s}\tr(M_{k}\e(\sigma))+\delta.\end{equation}
Such a pair of quantum states $(\rho,\sigma)$ is called a \emph{$(\epsilon,\delta)$-differentially private counterexample} of $\a$ within $\eta$. 

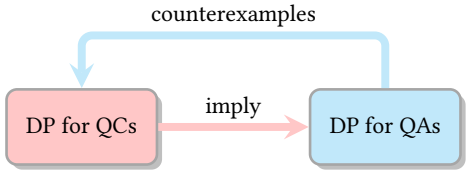
\begin{figure}[ht]
\centering
\begin{tikzpicture}[align=center,node distance=2cm]
  \node[mybox=colred] (QC) {DP for QCs};
  \node[mybox=colcyan,right=of QC] (QA) {DP for QAs};
  \draw[myarrow=colred] (QC) -- (QA);
  \draw[myarrow=colcyan] (QA.north) -- +(0,0.7) coordinate (VD1)
    -- (VD1-|QC) coordinate (VD2) -- (QC);
  \path (QC) -- node[mylabel=black,above]{imply} (QA);
  \path (VD1) -- node[mylabel=black,above]{counterexamples} (VD2);
\end{tikzpicture}
    \captionof{figure}{The relationship between the differential privacy (DP) for quantum circuits (QCs) and quantum algorithms (QAs)}
    \label{fig:workflow}
\end{figure}
As said before, the notion of differential privacy for (noisy) quantum circuits has been defined in the previous works~\cite{zhou2017differential,hirche2022quantum}. Using Definition~\ref{def:DP}, it can be reformulated as the following:

\begin{definition}[Differential Privacy for Quantum Circuits]\label{def:DP_superoperator}
  Suppose we are given a (noisy) quantum circuit $\e$ on a Hilbert space $\h$, a  distance metric $D(\cdot,\cdot)$ on $\dh$, and three small enough threshold values $\epsilon,\delta,\eta\geq 0$. Then $\e$ is said to be  $(\epsilon,\delta)$-differentially private within $\eta$  if for any quantum measurement $\{M_k\}_{k\in\o}$, the algorithm obtained from $\e$ by adding the measurement at the end, i.e. $(\e,\{M_k\}_{k\in\o})$, is $(\epsilon,\delta)$-differentially private within $\eta$.
\end{definition}


The relationship between differential privacy for quantum algorithms and quantum circuits can be visualized as Fig~\ref{fig:workflow}. More precisely, the differential privacy of a circuit $\e$ implies that of algorithm $(\e, \{M_k\}_{k\in\o})$ for any measurement $
\{M_k\}_{k\in\o}$. Conversely, for every measurement $\{M_k\}_{k\in\o}$, a counterexample of algorithm $(\e, \{M_k\}_{k\in\o})$ is also a counterexample of circuit $\e$. 


{\vskip 3pt}
\textit{\textbf{Choice of Distances}}: The reader should have noticed that the above definition of   differential privacy for quantum algorithms is similar to that for the classical datasets. But an intrinsic distinctness between them comes from different notions of neighboring relation. In the classical case, the state space of classical bits is discrete and two datasets are considered as  neighbors if they differ on a single bit. In the quantum case, two different neighboring relations for defining  quantum differential privacy have been adopted in the literature:
\begin{enumerate}
    \item As the state space of quantum bits is a continuum and thus uncountably infinite, a common way in the field of quantum computing  to define a neighboring relation is to introduce a distance $D$ that measures the closeness of two quantum states and set a bound $\eta$ on the distance. 
In~\cite{zhou2017differential} and several more recent papers~\cite{angrisani2022differential,hirche2022quantum}, trace distance is used to measure closeness (neighborhood). Trace distance is essentially a generalization of the total variation distance between probability distributions. It has been widely used by the quantum computation and quantum information community 
~\cite[Section 9.2]
{nielsen2010quantum}. Formally,
    for two quantum states $\rho,\sigma\in \dh$,
    \[D(\rho,\sigma)=\frac{1}{2}\tr(|\rho-\sigma|),\]
    where $|\rho-\sigma|=\Delta_++\Delta_-$ if $\rho-\sigma=\Delta_+-\Delta_-$ with $\tr(\Delta_+\Delta_{-})=0$ and $\Delta_{\pm}$ being positive semi-definite matrix.

{\vskip 3pt}
    \item [(2)]	In~\cite{aaronson2019gentle}, a way more similar to the setting of the classical database is introduced, where the neighboring relationship of two quantum states $\rho$ and $\sigma$ means that it’s possible to reach either $\sigma$ from $\rho$, or $\rho$ from $\sigma$, by performing a quantum operation (super-operator) on a single quantum bit only.
\end{enumerate}

Let us consider a simple  example about 2-qubit quantum states to further clarify  the difference between the above two approaches to defining quantum differential privacy. This example shows  that the definition through approach (1) is more suitable for the setting of \textit{noisy} quantum algorithms.
  \begin{example}\label{example_definition} 
        Given a 2-qubit state $\ket{0,1}$ (its mixed state form is $\rho=\ketbra{0}{0}\otimes \ketbra{1}{1}$). Under the bit-flip noise with probability $p_1$ (defined in Example~\ref{Exa:noise}) on the first qubit, the state $\rho$ will be changed to  
        \begin{equation*}
           \begin{aligned}
            \sigma_1&=\e_{X,p_1}(\ketbra{0}{0})\otimes \ketbra{1}{1}\\
            &=[(1-p_{1})\ketbra{0}{0}+p_1\ketbra{1}{1}]\otimes \ketbra{1}{1}.
           \end{aligned}
        \end{equation*}
According to  the above approach (2) 
$\rho$ and $\sigma_1$ are neighboring. They are also neighboring according to approach (1) if $p_1\leq \eta$.

However, the quantum noise cannot ideally be restricted to  a single qubit, but randomly effects on other qubits in the   system. In this case, if the second qubit of $\rho$ is simultaneously noisy under bit-flip with probability $p_2$, then the state $\rho$ will be further transferred to the following state:
        \begin{equation*}
           \begin{aligned}
            \sigma_2&=\e_{X,p_1}(\ketbra{0}{0})\otimes\e_{X,p_2} (\ketbra{1}{1})\\
            &=[(1-p_{1})\ketbra{0}{0}+p_1\ketbra{1}{1}]\otimes [(1-p_2)\ketbra{1}{1}+p_2\ketbra{0}{0}].
           \end{aligned}
        \end{equation*}
        It is easy to see that $\rho$ and $\sigma_{2}$ are not neighbors  under approach (2) even if the probability $p_2$ is extremely small, while they are neighboring under approach (1) provided $p_1+p_2-p_1p_2\leq \eta$.
    \end{example}

Targeting  the applications of detecting violations of differential privacy of quantum algorithms in  the current NISQ era where noises are unavoidable, we  follow  approach (1) in this paper. In particular,  $D(\cdot,\cdot)$ in Definition~\ref{def:DP} is chosen to be the trace distance, which is one of the more popular distances in  the quantum computation and information literature.   

\textit{\textbf{ Remark.}} As the trace distance of any two quantum states is within 1, the quantum differential privacy through approach  (1) 
implies that 
through approach (2) with $\eta=1$. However, the opposite direction does not hold. 

Furthermore, trace distance can maintain the neighboring relation between classical data vectors that differ by a single element. This allows quantum differential privacy guarantees on quantum states to be transferred back to guarantees on the privacy of the encoded classical data.
\begin{example}\label{Exa:angle_encoding}
    Consider two neighboring classical data vectors $\bar{v}$ and $\bar{w}$ that differ only in the $j^{th}$ element. Using angle encoding, they can be encoded into quantum states $\rho$ and $\sigma$, respectively. It can then be computed that:
\[D(\rho,\sigma) = \sqrt{1-\bra{0}R_j(v_j-w_j)\ket{0}\bra{0}R_j(w_j-v_j)\ket{0}}\]
where $R_j$ is the rotation gate used to encode the $j^{th}$ element of $\bar{v}$ and $\bar{w}$.
In particular, for binary vectors $\bar{v}, \bar{w} \in \{0,1\}^n$, the trace distance between the corresponding quantum states $\rho$ and $\sigma$ satisfies $D(\rho,\sigma) \leq \sin\frac{1}{2}$. This upper bound is achieved when $R_j$ is chosen to be rotations about the $x$- or $y$-axis, i.e., $R_x$ or $R_y$.  Therefore, by setting $\eta=\sin{\frac{1}{2}}$ in the definition of quantum differential privacy (Definition~\ref{def:DP}), the neighboring relation in classical data can be transferred to a relation between quantum states under trace distance. In other words, if two classical data vectors are considered neighbors because they differ by one element, then their angle-encoded quantum state representations will have trace distance $\sin{\frac{1}{2}}$. Subsequently, quantum differential privacy guarantees the privacy of the encoded classical data when used in quantum algorithms. By ensuring the quantum states satisfy differential privacy, the privacy of the original classical data is also ensured. 
\end{example}


{\vskip 3pt}

{\bf Noisy Post-processing:}
Similarly to the case of the classical computing~\cite{dwork2014algorithmic} and noisy quantum circuits~\cite{zhou2017differential}, the differential privacy for noiseless quantum algorithms is immune to noisy post-processing:  without additional knowledge about a noiseless quantum algorithm, any quantum noise applied on the output states  of a noiseless quantum algorithm does not increase privacy loss.
\begin{theorem}\label{Thm:post-processing}
    Let $\a=(\u,\{M_i\}_{i\in\o})$ be a noiseless quantum algorithm. Then for any (unknown) quantum noise represented by a super-operator $\f$, if $\a$ is $(\epsilon,\delta)$-differentially private, then $(\f\circ \u, \{M_{i}\}_{i\in \o})$ is also $(\epsilon,\delta)$-differentially private.  
\end{theorem}
However, the above theorem does not hold for a general noisy quantum algorithm $\mathcal{A}$ in the sense that unitary $\mathcal{U}$  is replaced by a (noisy) quantum operation modeled as a super-operator $\e$.   With the help of our main theorem (Theorem~\ref{Thm:verifier}) introduced later for differential privacy verification,  a concrete example showing this fact is provided as  Example~\ref{Appendix:example} at the end of the next section. 


{\vskip 3pt}
\textbf{Composition Theorem}: In order to handle larger quantum algorithms in a modular way, a series of composition theorems for  differential privacy of classical algorithms have been established~\cite{dwork2014algorithmic}. Some of them can be generalized into the quantum case. Given two quantum algorithms
$\a_k=(\e_{k},\{M_{k,j_k}\}_{j_k\in\o_k})$ $(k=1,2)$ , their parallel composition is $\a_{\s_1}\otimes\a_{\s_2}=(\e_{1}\otimes\e_{2},\{M_{1,\s_1}\otimes M_{2,\s_2},I-M_{1,\s_1}\otimes M_{2,\s_2}\})$ for some subsets $\s_k\subseteq \o_k(k=1,2)$, where $M_{k,\s_k}=\sum_{j_k\in\s_k}M_{k,j_k}$. Then we have: 
\begin{theorem}\label{Thm:composition}
For any subsets $\s_k\subseteq \o_k(k=1,2)$, 
\begin{enumerate}\item if $\a_k$ is $\epsilon_k$-differentially private within $\eta_k$ $(k=1,2)$, then  $\a_{\s_1}\otimes\a_{\s_2}$ is $(\epsilon_1+\epsilon_2)$-differentially private within $\eta_1\eta_2$;
\item if $\a_k$ is $(\epsilon_k,\delta_k)$-differentially private within $\eta_k$ $(k=1,2)$, then $\a_{\s_1}\otimes\a_{\s_2}$ is $(\epsilon_1+\epsilon_2,\delta_1+\delta_2)$-differentially private within $\eta_1\eta_2$.
\end{enumerate}
\end{theorem}

{\vskip 3pt}

\textit{\textbf{Remark.}} There are quite a few papers on the robustness of quantum machine learning~\cite{guan2021robustness,du2020quantum}. In these papers, the quantum robustness of quantum classiﬁer (which is mathematically a deterministic function) is the ability to make correct classiﬁcation with a small perturbation to a given input state (a local property), while quantum differential privacy ensures that a quantum algorithm (which is mathematically a randomized function) must behave similarly on all similar input states (a global property). Therefore, quantum differential privacy and robustness mainly differ on the studied functions and the property type. However, a deeper connection between quantum differential privacy and robustness may be built if we make some generalizations. In classical machine learning, the trade-off phenomenon of differential privacy and robustness has been found and several similarities of them have been reported if we can generalize the definition of robustness to randomized functions and consider Renyi-differential privacy~\cite{pinot2019unified}. However, this is still unclear in the quantum domain as the study of trustworthy quantum machine learning is at a very early age. We are interested in exploring this as the next step.


\section{Differential Privacy Verification}
In this section, we develop an algorithm for the differential privacy verification of quantum algorithms. Formally, the major problem concerned in this paper is the following:
\begin{problem}[Differential Privacy Verification Problem]\label{problem:DP}
Given a quantum algorithm $\a$ and $ 1\geq \epsilon,\delta,\eta\geq 0$, check whether or not $\a$ is $(\epsilon,\delta)$-differentially private within $\eta$. If not,  then (at least) one counterexample of quantum states $(\rho,\sigma)$ is provided.
\end{problem}

To solve this verification problem, we first find a necessary and sufficient condition for the differential privacy.  Specifically, we show that the differential privacy of a quantum algorithm can be characterized by a system of  inequalities. To this end, let us introduce several notations. For a positive semi-definite matrix $M$, we use 
$\lambda_{max}(M)$ and $\lambda_{min}(M)$ to denote  the maximum and minimum eigenvalues of  $M$, respectively. 
For a (noisy) quantum circuit modeled by a linear map $\e$ in the Kraus matrix form $\e=\{E_k\}_{k\in\k}$, the dual mapping of $\e$, denoted as $\e^\dagger$, is defined by 
\[\e^\dagger(M)=\sum_{k\in\k}E_k^\dagger M E_k \text{ for any positive semi-definite matrix } M.\]

\begin{theorem}[Sufficient and Necessary Condition]\label{Thm:verifier}
    Let $\a=(\e,\{M_{k}\}_{k\in\o})$ be a quantum algorithm. Then: \begin{enumerate}\item $\a$ is $(\epsilon,\delta)$-differentially private within $\eta$ if and only if \begin{equation}\label{privacy-in}\delta\geq \max_{\s\subseteq\o}\delta_{\s}\end{equation} where 
    \[\delta_{\s}=\eta\lambda_{\max}(M_{\s})-(e^\epsilon+\eta-1)\lambda_{\min}(M_{\s}),\]
     and matrix  $M_{\s}=\sum_{k\in\s}\e^\dagger(M_{k})$. 

    \item In particular, $\a$ is $\epsilon$-differentially private within $\eta$ if and only if  $\epsilon\geq \epsilon^*$, the optimal bound (minimum value) of $\epsilon$, where \[\epsilon^*=\ln[(\kappa^*-1)\eta+1] \quad\text{ and } \quad\kappa^*=\max_{\s\subseteq \o
    }\kappa(M_{\s}),\]
   $\kappa(M_{\s})=\frac{\lambda_{\max}(M_{\s})}{\lambda_{\min}(M_{\s})}$ is the condition number\footnote{ In numerical analysis, the condition number~\cite{higham1995condition} of  a matrix can be thought of both as a measure of the sensitivity of the solution of a linear system to perturbations in the data and as a measure of the sensitivity of the matrix inverse to perturbations in the matrix.} of matrix  $M_{\s}$, and if $\lambda_{\min}(M_\s)=0$, then  $\kappa(M_{\s})=+\infty$.\end{enumerate}
    \end{theorem}

  By the above theorem, we see that the verification problem (i.e. Problem \ref{problem:DP}) can be tackled by solving the system (\ref{privacy-in}) of inequalities. Consequently, it can be solved by computing the maximum and minimum eigenvalues (and their eigenvectors) of  positive semi-definite matrix $M_{\s}$. In particular, for the case of $\epsilon$-differential privacy, we have:
  \begin{enumerate}
      \item [(1)] the maximum value $1\leq \kappa^*\leq +\infty$ of condition numbers of $M_{\s}$ over $\s\subseteq \o$ measures the $\epsilon$-differential privacy of  noisy quantum algorithm $\a=(\e,\{M_{k}\}_{k\in\o})$ (for fixed $\eta$). For the extreme cases, 
      \begin{enumerate}
          \item [(i)] if $k^*=1$, then $\epsilon^*=0$, and $\a$ is $\epsilon$-differentially private for any $\epsilon\geq 0$;
          \item [(ii)] if $k^*=+\infty$,  then $\epsilon^*=+\infty$, and $\a$ is not $\epsilon$-differentially private for any $\epsilon\geq 0$.
      \end{enumerate}
      In the following evaluation (Section~\ref{Sec:evaluation}), we will compute $\kappa^*$ for diverse noisy quantum algorithms with different noise levels on quantum circuits  to show  that quantum differential privacy can benefit from the  quantum noises on quantum circuits.
      \item [(2)]  we can  characterize the $\epsilon$-diﬀerential privacy of a noisy quantum algorithm for diﬀerent values of $\eta$, i.e., the optimal bound $\epsilon*$  can be regarded as a function $\epsilon^*(\cdot)$ of $\eta$ as follows:
      \[\epsilon^*(\eta)=\ln[(\kappa^*-1)\eta+1] \ \text{ where } \ \kappa^*\geq 1.\]
      As we can see from the above equation, the value of $\epsilon^*$ logarithmically increases with $\eta$. This reveals that  as the quantum noise level on input states increases, the differential privacy increases because $\eta$ can measure the noisy neighboring relation of the input states effected by the quantum noises, which has been  illustrated after  Definition~\ref{def:DP} and by Example~\ref{example_definition}. This finding provides the theoretical guarantee that adding noises to input states is a way to improve the differential privacy of quantum algorithms.
  \end{enumerate} 
  In summary, quantum differential privacy can benefit from the quantum noise on either quantum circuits or  input states.


    Furthermore, we are able to give a characterization of differential privacy counterexamples: 
    
    \begin{theorem}[Counterexamples]\label{Thm:generator} If $\a$ is not $(\epsilon,\delta)$-differentially private within $\eta$, then for any $\s\subseteq\o$ with $\delta<\delta_{\s}$ (defined in Theorem~\ref{Thm:verifier}),
    any pair of quantum states $(\gamma,\phi)$ of the form:
    \[\gamma=\eta{\psi}+(1-\eta){\phi}\qquad \phi=\ketbra{\phi}{\phi}\]
    is a $(\epsilon,\delta)$-differential privacy  counterexample within $\eta$, 
     where $\psi=\ketbra{\psi}{\psi}$, and $\ket{\psi}$ and $\ket{\phi}$ are normalized eigenvectors of $M_{\s}$ (defined in Theorem~\ref{Thm:verifier}) corresponding to the maximum and minimum eigenvalues, respectively.
\end{theorem}

Now we are ready to provide an example showing that Theorem~\ref{Thm:post-processing} does not hold for noisy quantum algorithms. This example also demonstrates the method for solving the verification problem (Problem~\ref{problem:DP}) using Theorem~\ref{Thm:verifier} and~\ref{Thm:generator}. 

\begin{example}\label{Appendix:example}
Let $\h$ be a 2-qubit Hilbert space, i.e., \[\h = \mathrm{span}\{\ket{0,0},\ket{0,1},\ket{1,0},\ket{1,1}\},\] and $\a = (\e, \{M_0,M_1\})$ be a noisy quantum algorithm on $\h$, where $\e$  is not a unitary but a super-operator with the Kraus matrix form $\e=\{E_{i}\}_{i=1}^4$ with
\begin{equation*}
    \begin{aligned}
        &E_1 = \frac{1}{\sqrt{3}}(\ket{0,0}+\ket{1,0}+\ket{1,1})\bra{0,0} \\
        &E_2=\frac{1}{\sqrt{3}}(\ket{0,1}+\ket{1,0}+\ket{1,1})\bra{0,1}\\
        &E_3=\frac{1}{\sqrt{6}}(\ket{0,0}+\ket{0,1}+{2}\ket{1,0})\bra{1,0}\\
         &E_4=\frac{1}{\sqrt{6}}(\ket{0,0}+\ket{0,1}+{2}\ket{1,1})\bra{1,1}
    \end{aligned}
\end{equation*} and measurement operators \[M_0 = \ketbra{0,0}{0,0}+\ketbra{0,1}{0,1}\quad M_1 = \ketbra{1,0}{1,0}+\ketbra{1,1}{1,1}.\]
It can be calculated that 
\begin{equation*}
    \begin{aligned}
     &\e^{\dagger}(M_0) = \frac{1}{3} (\ketbra{0,0}{0,0}+\ketbra{0,1}{0,1}+\ketbra{1,0}{1,0}+\ketbra{1,1}{1,1})\\
     &\e^{\dagger}(M_1)=\frac{2}{3} (\ketbra{0,0}{0,0}+\ketbra{0,1}{0,1}+\ketbra{1,0}{1,0}+\ketbra{1,1}{1,1}).
    \end{aligned}
\end{equation*} Then 
\begin{equation*}
    \begin{aligned}
    &\lambda_{\max}(\e^{\dagger}(M_0+M_1)) = \lambda_{\min}(\e^{\dagger}(M_0+M_1)) = 1\\
    &\lambda_{\max}(\e^{\dagger}(M_0)) = \lambda_{\min}(\e^{\dagger}(M_0)) = \frac{1}{3}\\
    &\lambda_{\max}(\e^{\dagger}(M_1)) = \lambda_{\min}(\e^{\dagger}(M_1)) = \frac{2}{3}.
    \end{aligned}
\end{equation*}
Consequently, $\kappa^*=1$ implies $\epsilon^*=0$ by Theorem~\ref{Thm:verifier} and then $\a$ is $\epsilon$-differentially private for any  $\epsilon\geq 0$.

However, if we choose a quantum noise represented by the following  super-operator 
\[\f=\left\{\ketbra{0,0}{0,0}, \ketbra{1,0}{0,1}, \ketbra{1,0}{1,0},\ketbra{1,1}{1,1}\right\}\]
such that 
\begin{equation*}
    \begin{aligned}
        (\f\circ\e)^{\dagger}(M_0) &= \e^{\dagger}(\f^{\dagger}(M_0)) \\
        &= \frac{1}{6}(2\ketbra{0,0}{0,0} + \ketbra{1,0}{1,0} + \ketbra{1,1}{1,1}).
    \end{aligned}
\end{equation*} 
Then \[\lambda_{\max}((\f\circ\e)^{\dagger}(M_0)) = \frac{1}{3}   \qquad \lambda_{\min}((\f\circ\e)^{\dagger}(M_0)) = 0\]
with normalized eigenvectors $\ket{0,0}$ and $\ket{0,1}$, respectively. Thus  $\kappa^*=+\infty$ implies $\epsilon^*=+\infty$ by Theorem~\ref{Thm:verifier}. Subsequently, the  noisy quantum algorithm $\a'=(\f\circ\e, \{M_0, M_1\})$ is not $\epsilon$-differentially private for any $\epsilon\geq 0$. Furthermore, in this case, by Theorem~\ref{Thm:generator}, $(\gamma, \phi)$ is a $\epsilon$-differential privacy counterexample of the algorithm for any $\epsilon\geq 0$, where
\[\gamma=\eta\ketbra{0,0}{0,0}+(1-\eta){\ketbra{0,1}{0,1}}\qquad \phi=\ketbra{0,1}{0,1}.\]
\end{example}

\subsection{Differential Privacy Verification Algorithm}\label{sec_algorithms}

Theorems \ref{Thm:verifier} and  \ref{Thm:generator} provide a theoretical basis for developing algorithms for verification and violation detection of quantum differential privacy.  
Now we are ready to present them.  Algorithm~\ref{Algorithm} is designed  for  verifying the $(\epsilon,\delta)$-differential privacy for (noisy) quantum algorithms. For estimating parameter $\epsilon$ in the $\epsilon$-differential privacy, Algorithm~\ref{Algorithm:kappa} is developed to compute the maximum condition number $\kappa^*$ (with a counterexample)  as in  Theorem~\ref{Thm:verifier}. By calling Algorithm~\ref{Algorithm:kappa}, an alternative way for verifying $\epsilon$-differential privacy is obtained as Algorithm~\ref{Algorithm:epsilon}.

\begin{algorithm}[ht]

\caption{DP($\a,\epsilon,\delta,\eta$)}
\label{Algorithm}
    \begin{algorithmic}[1]
    \Require A quantum algorithm $\a=(\e=\{E_j\}_{j\in\j},\{M_{k}\}_{k\in\o})$  on a Hilbert space $\h$ with dimension $2^n$, and real numbers $ \epsilon,\delta,\eta\geq 0$.
    \Ensure \True{} indicates $\a$ is $(\epsilon,\delta)$-differentially private within $\eta$ or \False{} with a counterexample $(\rho,\sigma)$ indicates  $\a$ is not $(\epsilon,\delta)$-differentially private within $\eta$.
        \ForAll{$k\in \o$}\label{line:1}
    \State $W_k=\e^\dagger(M_{k})=\sum_{j\in\j}E_j^\dagger M_k E_j $\label{line:2}
    \EndFor
    \State $\delta^*=0$, $\s^*=\emptyset$ be an empty set and $M_{\s^*}={\bf 0}$, zero matrix.
    \ForAll{$\s\subseteq\o$}
    \State $M_{\s}=\sum_{k\in \s}W_k$ and $\delta_{\s}=\eta\lambda_{\max}(M_\s)-(e^\epsilon+\eta-1)\lambda_{\min}(M_\s)$\label{line:6}
    \If{$\delta_{\s}>\delta^*$}
    \State $\delta^*=\delta_{\s}$, $\s^*=\s$ and $M_{\s^*}=M_{\s}$
    \EndIf
    \EndFor
    \If{$\delta\geq \delta^*$}
    \State \Return \True{}\label{line:12}
    \Else
    \State  $\ket{\psi}$ and $\ket{\phi}$ are obtained from two normalized eigenvectors corresponding to the maximum and minimum eigenvalues of $M_{\s^*}$, respectively.\label{line:14}
    \State\Return \False{} and $(\eta\psi+(1-\eta)\phi,\phi)$\label{line:15}
    \EndIf
    \end{algorithmic}  
\end{algorithm}

In the following, we analyze the correctness and complexity of  Algorithm~\ref{Algorithm}. Those of Algorithms~\ref{Algorithm:kappa} and \ref{Algorithm:epsilon} can be derived in a similar way.

{\bf Correctness:} Algorithm~\ref{Algorithm} consists of two components --- a verifier (Lines~\ref{line:1}-\ref{line:12}) and a counterexample generator (Lines~\ref{line:14}-\ref{line:15}). Following the verification procedure in the first part of Theorem~\ref{Thm:verifier}, the verifier is designed to check whether or not a quantum algorithm is $(\epsilon,\delta)$-differentially private within $\eta$. The counterexample generator is constructed using  Theorem~\ref{Thm:generator} asserting that $(\eta\psi+(1-\eta)\phi,\phi)$ is a $(\epsilon,\delta)$-differential privacy counterexample if there is a subset $\s\subseteq \o$, i.e., $\s^*$ in the algorithm,  such that $\delta^*=\delta_{\s^*}>\delta$, where $\ket{\psi}$ and $\ket{\phi}$ are normalized eigenvectors of $M_{\s^*}$ (defined in Theorem~\ref{Thm:verifier}) corresponding to the maximum and minimum eigenvalues, respectively.

\begin{algorithm}[ht]
\caption{$\text{DP}_\kappa(\a)$}
\label{Algorithm:kappa}
    \begin{algorithmic}[1]
    \Require  A quantum algorithm $\a=(\e=\{E_j\}_{j\in\j},\{M_{k}\}_{k\in\o})$  on a Hilbert space $\h$ with dimension $2^n$.
    \Ensure The maximum condition number $\kappa^*$  and a counterexample as in Theorems~\ref{Thm:verifier} and~\ref{Thm:generator}, respectively. 
    \ForAll{$i\in \o$}
    \State $W_i=\e^\dagger(M_{k})=\sum_{j\in\j}E_j^\dagger M_{k} E_j $
    \EndFor
    \State $\kappa^*=0$, $\s^*=\emptyset$ be an empty set and $M_{\s^*}={\bf 0}$, the zero matrix.
    \ForAll{$\s\subseteq\o$}
    \State  $\kappa(M_{\s})=\frac{\lambda_{\max}(M_{\s})}{\lambda_{\min}(M_{\s})}$ for $M_{\s}=\sum_{k\in \s}W_k$
    \If{$\kappa(M_{\s})>\kappa^*$}
    \State $\kappa^*=\kappa(M_\s)$, $\s^*=\s$ and $M_{\s^*}=M_\s$
    \EndIf
    \EndFor
    \State $\ket{\psi}$ and $\ket{\phi}$ are obtained from two normalized eigenvectors corresponding to the maximum and minimum eigenvalues of $M_{\s^*}$, respectively.
    \State\Return $\kappa^*$ and $(\eta\psi+(1-\eta)\phi,\phi)$
    \end{algorithmic}
\end{algorithm}

{\bf Complexity:}
 The complexity of Algorithm~\ref{Algorithm} mainly attributes to the calculations in Lines~\ref{line:2},~\ref{line:6} and~\ref{line:14}. In Line~\ref{line:2}, computing   $W_k=\sum_{j\in\j}E_j^\dagger M_k E_j $ for each $k\in\o$ needs $O(2^{5n})$ as the multiplication of $2^n\times 2^n$ matrices needs  $O(2^{3n})$ operations, and the number $|\j|$ of the Kraus operators $\{E_j\}_{j\in\j}$ of $\e$ can be at most $2^{2n}$~\cite[Chapter 2.2]{wolf2012quantum}; In Line~\ref{line:6}, calculating  $\sum_{k\in \s}W_k$  and its maximum and minimum eigenvalues  (and the corresponding eigenvectors for $\s=\s^*$ in Line~\ref{line:14}) for each $A \subseteq \o$ costs $O(2^{|\o|}|\o|2^{2n})$ since the number of subsets of $\o$ is $2^{|\o|}$, $|\s|\leq |\o|$ for any $\s\subseteq \o$ and computing maximum and minimum eigenvalues with corresponding eigenvectors of $2^n\times 2^n$ matrix by basic power method~\cite{doi:10.1137/1.9780898719581} costs $O(2^{2n})$. Therefore, the total complexity of Algorithm~\ref{Algorithm} is $O(2^{5n}+2^{|\o|}|\o|2^{2n})$.

\begin{algorithm}[ht]

\caption{$\text{DP}_\epsilon(\a,\epsilon,\eta$)}
\label{Algorithm:epsilon}
    \begin{algorithmic}[1]
    \Require A quantum algorithm $\a=(\e=\{E_j\}_{j\in\j},\{M_{k}\}_{k\in\o})$  on a Hilbert space $\h$ with dimension $2^n$, and real numbers $\epsilon,\eta\geq 0$.
    \Ensure \True{} indicates $\a$ is $\epsilon$-differentially private within $\eta$ or \False{} with a counterexample $(\rho,\sigma)$ indicates  $\a$ is not $\epsilon$-differentially private within $\eta$.
    \State\label{line:call} $(\kappa^*,(\eta\psi+(1-\eta)\phi,\phi))=\text{DP}_{\kappa}(\a)$\Comment{Call Algorithm~\ref{Algorithm:kappa}}
    \If{$\epsilon\geq\ln[(\kappa^*-1)\eta+1$]}\label{line:compare}
    \State \Return \True{}
    \Else
    \State\Return \False{} and $(\eta\psi+(1-\eta)\phi,\phi)$
    \EndIf\label{line:the_end}
    \end{algorithmic}    
\end{algorithm}

The above calculations are also the main computational cost in Algorithms~\ref{Algorithm:kappa} and \ref{Algorithm:epsilon}, so the two algorithms   share the same complexity with Algorithm~\ref{Algorithm}.
 
\begin{theorem}\label{Thm:complexities}
The worst-case complexities of Algorithms~\ref{Algorithm}, ~\ref{Algorithm:kappa} and~\ref{Algorithm:epsilon}, are all $O(2^{5n}+2^{|\o|}|\o|2^{2n})$, where $n$ is the number of the qubits in quantum algorithms and $|\o|$ is the number of the measurement outcome set $\o$.
\end{theorem}

\textit{\textbf{ Remark.}} As we can see in Theorem~\ref{Thm:complexities}, the main limitation of our verification algorithms is the exponential complexity in the number of qubits. To overcome this scaling issue, we apply optimization techniques based on tensor networks to capture the locality and regularity of quantum circuits. This allows us to speed up the calculations involved in verification. As a result, we are able to verify quantum algorithms with up to 21 qubits, as shown in the later experimental section. 

Further improving the scalability of verified qubits is possible by adapting classical approximation methods to the quantum domain, as they have successfully analyzed large-scale classical machine learning algorithms~\cite{albarghouthi2021introduction}. Two promising techniques are:
\begin{itemize}
\item Abstraction-based approximation using abstract interpretation provides over-approximations of concrete program semantics. If a property holds for the abstracted version, it also holds for the original. This technique has boosted verification scalability for classical neural network robustness ~\cite{elboher2020abstraction} and correctness of quantum circuits up to 300 qubits ~\cite{yu2021quantum}.

\item Bound-based approximation derives efficiently computable bounds on algorithm properties. If the algorithm satisfies the bound, it satisfies the property, but the converse is unknown. This has enabled robustness verification for large-scale classical neural networks ~\cite{lin2019robustness} and quantum classifiers ~\cite{guan2021robustness}. 
\end{itemize}
These approximation methods trade off formal guarantees for scalability in verifying algorithm properties. Since quantum algorithms rely on quantum circuits, we can follow similar approaches~\cite{yu2021quantum,guan2021robustness} to improve the scalability of verifying quantum differential privacy.

\section{Evaluation}\label{Sec:evaluation}


In this section, we evaluate the effectiveness and efficiency of our Algorithms on noisy quantum algorithms.
{\vskip 3pt}
{\bf Implementation:} We implemented our Algorithms on the top of Google's Python software libraries: \texttt{Cirq} for writing and manipulating quantum circuits, \texttt{TensorNetwork} for converting quantum circuits to tensor networks.
Our implementation  supports  circuit models not only written in \texttt{Cirq} but also  imported from IBM's \texttt{Qiskit}, and accepts quantum machine learning models from both TensorFlow Quantum and TorchQuantum.

{\vskip 3pt}
{\bf Optimization Techniques:} We convert quantum circuits into tensor networks, which is a data structure exploiting the regularity and locality contained in quantum circuits, while matrix representation cannot.
The multiplication of matrices in our algorithm is transformed into the contraction of tensor networks.
For the tensor network of a quantum circuit, its complexity of contraction is $T^{O(1)}\exp[O(qd)]$~\cite{markov2008simulating}, where 
$T$ is the number of gates (tensors), $d$ is the depth in the circuit (tensor network) and $q$ is the number of allowed interacting qubits, i.e., the maximal number of qubits (legs of a tensor)  a gate applying on.
So we can avoid the exponential complexity of the number $n$ of qubits with the cost of introducing exponential complexity of $qd$, where $d$ and $q$ capture the regularity and locality of the quantum circuit, respectively.
Usually, $q=2$  for controlled gates, and then the complexity turns out to be $T^{O(1)}\exp[O(d)]$.
Even though the worst-case (presented in the complexity) is exponential on $d$,  there are a bulk of efficient algorithms to implement tensor network contraction for practical large-size quantum circuits. As a result, we can handle (up to) 21 qubits in the verification experiments avoiding the worst-case complexities of our algorithms presented in Theorem~\ref{Thm:complexities} that the time cost is exponential with the number $n$ of qubits. For more details on tensor networks representing quantum circuits, we refer to~\cite{bridgeman2017hand}.
 

{\vskip 3pt}
{\bf Platform:} We conducted our experiments on a machine with Intel Xeon Platinum 8153 @ 2.00GHz × 256 Cores, 2048 GB Memory, and no dedicated GPU, running Centos 7.7.1908.

{\vskip 3pt}
{\bf Benchmarks:}
To evaluate the efficiency and utility of our implementation, we test our algorithms on four groups of examples, including quantum approximate
optimization algorithms, quantum supremacy algorithms,  variational quantum eigensolver algorithms and quantum machine learning models (well-trained algorithms) for solving classical tasks with angle encoding introduced in Section~\ref{sec:preliminary}. All of them have been implemented on current NISQ computers.


\subsection{Quantum Approximate Optimization Algorithms}
The Quantum Approximate Optimization Algorithm (QAOA) is a quantum algorithm for producing approximate solutions for combinatorial optimization problems~\cite{farhi2014quantum}. Fig.~\ref{fig:qaoa}. shows a 2-qubit example of QAOA circuit. In our experiment, we use the circuit for hardware grid problems in ~\cite{harrigan2021quantum} generated from code in Recirq~\cite{quantum_ai_team_and_collaborators_2020_4091470}. Circuit name \emph{qaoa\_$D$} represents such a QAOA circuit with $D$ connected qubits on  Google's \textit{Sycarmore} quantum processor. 

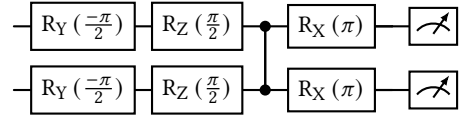
\begin{figure}[!ht]
    \centering
    \begin{quantikz}[row sep=2.2mm, column sep=2.3mm,]
        & \gate{\mathrm{R_Y}\,(\mathrm{\frac{-\pi}{2}})} & \gate{\mathrm{R_Z}\,(\mathrm{\frac{\pi}{2}})} & \ctrl{1} & \gate{\mathrm{R_X}\,(\mathrm{\pi})} & \qw & \meter{}\\
        & \gate{\mathrm{R_Y}\,(\mathrm{\frac{-\pi}{2}})} & \gate{\mathrm{R_Z}\,(\mathrm{\frac{\pi}{2}})} & \ctrl{-1} & \gate{\mathrm{R_X}\,(\mathrm{\pi})} & \qw & \meter{}
        \end{quantikz}
        \vskip -5pt
        \caption{A 2-qubit QAOA circuit.}
        \label{fig:qaoa}
\end{figure}

\subsection{Variational Quantum Eigensolver Algorithms}
The circuit of Variational Quantum Eigensolver (VQE) Algorithms comes from the experiments in \cite{google2020hartree}, which uses  Google's \textit{Sycarmore} quantum processor to calculate the binding energy of hydrogen chains. Fig.~\ref{fig:vqe}. shows an 8-qubit basis rotation circuit for $H_8$ used in the VQE algorithm. In our experiment, VQE circuit is obtained from Recirq and named \emph{hf\_$E$} with $E$ being the number of qubits.

\begin{figure}[!ht]
    \centering
    \subcaptionbox*{}[\linewidth]{
     \scalebox{0.75}{\tikzinput{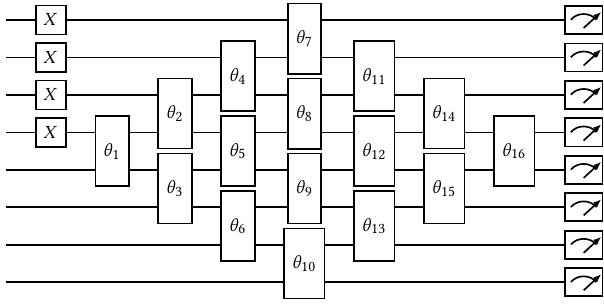}{tikz/hatree_fock}}}
         \vskip -10pt
 \subcaptionbox*{}[\linewidth]{
         \centering
         \scalebox{0.8}{\begin{quantikz}[row sep=2.2mm, column sep=2.3mm,]
             & \gate[2]{\theta} &\qw\\
             & \qw &\qw
             \end{quantikz}}
             :=
             \scalebox{0.8}{\begin{quantikz}[row sep=2.2mm, column sep=2.3mm,]
             & \gate[2]{\sqrt{iSWAP}}  & \gate{e^{i\theta Z/2}}                &\gate[2]{\sqrt{iSWAP}} &\qw                   & \qw   \\
             &                         & \gate{e^{-i(\theta + \pi) Z/2}}       &\qw                    &\gate{e^{-i \pi Z/2}} &\qw   
         \end{quantikz}}
         }
    \vskip -20pt
    \caption{An 8-qubit Hatree-Fock VQE circuit.}
    \label{fig:vqe}
\end{figure}

\subsection{Quantum Supremacy Algorithms}
The quantum supremacy algorithm  includes a specific random circuit designed to show the quantum supremacy on grid qubits~\cite{boixo2018characterizing}. In general, the circuit contains a number of cycles consisting of 1-qubit ($X^{1/2}$, $Y^{1/2}$ and $T$ gate) and 2-qubit quantum gates (CZ gate). The 2-qubit gates are implemented in a specific order according to the topology of the grid qubits, where each qubit in the middle of the circuit is connected to four qubits, and the qubits on the edges and corners are connected to three and two qubits, respectively. The circuit is implemented on Google \emph{Sycarmore} quantum processor to show the quantum supremacy~\cite{arute2019quantum}. In our experiment, the circuits are named by \emph{inst\_$A\times B$\_$C$}, representing an $(A \times B)$-qubit circuit with depth $C$. See Fig.~\ref{examples:b} for an example of $2\times 2$-qubit quantum supremacy algorithms.

\subsection{Quantum Machine Learning Models}
There are two frameworks, TensorFlow Quantum and TorchQuantum, which are based on famous machine learning platforms --- TensorFlow and Pytorch, respectively, for training and designing quantum machine learning models.
TensorFlow Quantum uses \texttt{Cirq} to manipulate quantum circuits, and so does our implementation.
TorchQuantum supports the conversion of models into quantum circuits described by \texttt{Qiskit}, which can also be converted to \texttt{Cirq} by our implementation.
Thus, our implementation is fully compatible with both TensorFlow Quantum and TorchQuantum.

We collect two quantum machine learning models using Tensorflow Quantum for financial tasks, as described in~\cite{guan2022verifying}. All classical financial data are encoded into quantum states using the angle encoding method introduced in Section~\ref{sec:preliminary}.
\begin{itemize}
    \item The model named \emph{GC\_9}, trained on public German credit card dataset~\cite{Dua:2019}, is used to classify whether a person has a good credit.

    \item The model named \emph{AI\_8}, trained on public adult income dataset~\cite{dice2020}, is used to predict whether an individual's income exceeds $\$50,000/\text{year}$ or not.
\end{itemize}
Additionally, we train a model called \emph{EC\_9} to detect fraudulent credit card transactions. The model is trained on a dataset of European cardholder transactions~\cite{ec}.

Furthermore, we evaluate two quantum machine learning models from the TorchQuantum library paper~\cite{wang2022quantumnas}, which introduces a PyTorch framework for hybrid quantum-classical machine learning.
\begin{itemize}
    \item The model \emph{MNIST\_10}, trained on MNIST~\cite{lecun-mnisthandwrittendigit-2010}, is used to classify handwritten digits.
    \item The model \emph{Fashion\_4}, trained on Fashion MNIST~\cite{xiao2017/online}, is used to classify fashion images.
\end{itemize}
As before, handwritten digits and fashion images are encoded into quantum states via angle encoding.

\subsection{Differential Privacy Verification and Analysis}
{\bf Verification Algorithms:} As shown in~Theorem~\ref{Thm:complexities}, the complexities of our Algorithms~\ref{Algorithm}, ~\ref{Algorithm:kappa} and ~\ref{Algorithm:epsilon} are the same, so for convenience, we only test the implementation of Algorithm~\ref{Algorithm:kappa} since it only requires quantum algorithms as the input without factors $\epsilon,\delta,\eta$ for verifying differential privacy.
In addition, to demonstrate the noisy  impact on quantum algorithms in the NISQ era, we add two types of quantum noises --- depolarizing and bit flip with different levels of probability to each qubit in all circuits of our examples.
Then we run Algorithm~\ref{Algorithm:kappa} to evaluate the maximum condition number $\kappa^*$ of all examples.
The evaluation results are summarized in Tables~\ref{table:qaoa}-~\ref{table:qml}.
It can be seen that the higher level of noise’s probability, the smaller value of the maximum condition number $\kappa^*$.
So similarly to protecting classical differential privacy by adding noises, quantum algorithms also benefit from quantum noises on circuits in terms of quantum differential privacy.  It is worth noting that in all experiments, we also obtain differential privacy counterexamples by Algorithm~\ref{Algorithm:kappa} at the running time presented in the tables, but as they are large-size (up to $2^{21}\times 2^{21}$) matrices, we do not show them here.

\begin{table}[!ht]
    \centering
    \caption{Experimental results of the maximum condition number $\kappa^*$ on \textit{Quantum Approximate Optimization
    Algorithms} with different noise levels.}\label{table:qaoa}
    \small
    \vskip -10pt
    \begin{tabular}[c]{cccccc}
        \toprule \bfseries Circuit & \bfseries \#Qubits & \bfseries Noise Type & $\bm{p}$ & $\bm{\kappa^{*}}$ & \bfseries Time (s) \\
        \hline
        \multirow{4}{*}{\emph{qaoa\_20}} & \multirow{4}{*}{20} & \multirow{2}{*}{depolarizing} & 0.01 & 62.39 & 285.80 \\
         & & & 0.001 & 747.21 & 312.38 \\
         & & \multirow{2}{*}{bit flip} & 0.01 & 88.53 & 220.73 \\
         & & & 0.001 & 852.94 & 216.86 \\
        \hline
        \multirow{4}{*}{\emph{qaoa\_21}} & \multirow{4}{*}{21} & \multirow{2}{*}{depolarizing} & 0.01 & 97.58 & 644.51 \\
         & & & 0.001 & 1032.48 & 514.83 \\
         & & \multirow{2}{*}{bit flip} & 0.01 & 91.27 & 583.85 \\
         & & & 0.001 & 923.85 & 594.24 \\
        \toprule
    \end{tabular}
\end{table}

\begin{table}[!ht]
    \centering
    \caption{Experimental results of the maximum condition number $\kappa^*$ on \textit{Variational Quantum Eigensolver
    Algorithms} with different noise levels.}\label{table:vqe}
    \small
    \vskip -10pt
    \begin{tabular}[c]{cccccc}
        \toprule \bfseries Circuit & \bfseries \#Qubits & \bfseries Noise Type & $\bm{p}$ & $\bm{\kappa^{*}}$ & \bfseries Time (s) \\
        \hline
        \multirow{4}{*}{\emph{hf\_8}} & \multirow{4}{*}{8} & \multirow{2}{*}{depolarizing} & 0.01 & 135.50 & 277.37 \\
         & & & 0.001 & 1412.58 & 212.06 \\
         & & \multirow{2}{*}{bit flip} & 0.01 & 98.39 & 248.36 \\
         & & & 0.001 & 991.73 & 259.37 \\
        \hline
        \multirow{4}{*}{\emph{hf\_10}} & \multirow{4}{*}{10} & \multirow{2}{*}{depolarizing} & 0.01 & 132.21 & 477.70 \\
         & & & 0.001 & 1423.75 & 482.10 \\
         & & \multirow{2}{*}{bit flip} & 0.01 & 97.64 & 409.25 \\
         & & & 0.001 & 988.26 & 427.58 \\
        \hline
        \multirow{4}{*}{\emph{hf\_12}} & \multirow{4}{*}{12} & \multirow{2}{*}{depolarizing} & 0.01 & 140.58 & 955.22 \\
         & & & 0.001 & 1438.94 & 962.34 \\
         & & \multirow{2}{*}{bit flip} & 0.01 & 95.27 & 890.26 \\
         & & & 0.001 & 978.87 & 816.83 \\
        \toprule
    \end{tabular}
\end{table}

\begin{table}[!ht]
    \centering
    \caption{Experimental results of the maximum condition number $\kappa^*$ on \textit{Quantum Supremacy Algorithms} with different noise levels.}\label{table:inst}
    \small
    \vskip -10pt
    \begin{tabular}[c]{cccccc}
        \toprule \bfseries Circuit & \bfseries \#Qubits & \bfseries Noise Type & $\bm{p}$ & $\bm{\kappa^{*}}$ & \bfseries Time (s) \\
        \hline
        \multirow{4}{*}{\emph{inst\_4x4\_10}} & \multirow{4}{*}{16} & \multirow{2}{*}{depolarizing} & 0.01 & 59.67 & 254.05 \\
         & & & 0.001 & 748.51 & 247.42 \\
         & & \multirow{2}{*}{bit flip} & 0.01 & 82.39 & 207.39 \\
         & & & 0.001 & 901.74 & 213.18 \\
        \hline
        \multirow{4}{*}{\emph{inst\_4x5\_10}} & \multirow{4}{*}{20} & \multirow{2}{*}{depolarizing} & 0.01 & 62.05 & 13176.98 \\
         & & & 0.001 & 823.85 & 7493.24 \\
         & & \multirow{2}{*}{bit flip} & 0.01 & 88.72 & 8120.35 \\
         & & & 0.001 & 918.87 & 8203.71 \\
        \toprule
    \end{tabular}
\end{table}

\begin{table}[!ht]
    \centering
    \caption{Experimental results of the maximum condition number $\kappa^*$ on various \textit{Quantum Machine Learning Models} with different noise levels.}\label{table:qml}
    {\vskip -10pt}
    \small
    \begin{tabular}[c]{cccccc}
        \toprule \bfseries Circuit & \bfseries \#Qubits & \bfseries Noise Type & $\bm{p}$ & $\bm{\kappa^{*}}$ & \bfseries Time (s) \\
        \hline
        \multirow{4}{*}{\emph{EC\_9}}
        & \multirow{4}{*}{9} & \multirow{2}{*}{depolarizing} & 0.01    & 3.370     & 5.49 \\
         & & & 0.001   & 32.199    & 3.61 \\
         & & \multirow{2}{*}{bit flip} & 0.01    & 3.144     & 3.95 \\
         & & & 0.001   & 29.466    & 3.85 \\
        \hline
        \multirow{4}{*}{\emph{GC\_9}} & \multirow{4}{*}{9} & \multirow{2}{*}{depolarizing} & 0.01 & 4.236 & 5.12 \\
         & & & 0.001 & 41.077 & 3.92 \\
         & & \multirow{2}{*}{bit flip} & 0.01 & 4.458 & 4.09 \\
         & & & 0.001 & 42.862 & 3.80 \\
        \hline
        \multirow{4}{*}{\emph{AI\_8}} & \multirow{4}{*}{8} & \multirow{2}{*}{depolarizing} & 0.01 & 4.380 & 3.54 \\
         & & & 0.001 & 42.258 & 2.58 \\
         & & \multirow{2}{*}{bit flip} & 0.01 & 5.025 & 2.20 \\
         & & & 0.001 & 50.108 & 2.44 \\
        \hline
        \multirow{4}{*}{\emph{Mnist\_10}} & \multirow{4}{*}{10} & \multirow{2}{*}{depolarizing} & 0.01 & 1.170 & 18.90 \\
         & & & 0.001 & 7.241 & 17.44 \\
         & & \multirow{2}{*}{bit flip} & 0.01 & 1.132 & 17.39 \\
         & & & 0.001 & 6.677 & 17.14 \\
        \hline
        \multirow{4}{*}{\emph{Fashion\_4}} & \multirow{4}{*}{4} & \multirow{2}{*}{depolarizing} & 0.01 & 1.052 & 3.29 \\
         & & & 0.001 & 5.398 & 3.18 \\
         & & \multirow{2}{*}{bit flip} & 0.01 & 1.057 & 3.26 \\
         & & & 0.001 & 5.635 & 3.27 \\
        \toprule
    \end{tabular}
\end{table}

 
\begin{figure}[!ht]
    \centering
    
    \subcaptionbox{The optimal bound function $\epsilon^*(\eta)$ for \emph{hf\_12}.}[\linewidth]{
        \tikzinput{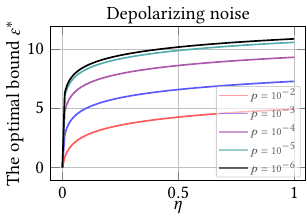}{tikz/hf_dp.tex}
    }
    
    \subcaptionbox{The optimal bound function $\epsilon^*(\eta)$ for \emph{hf\_12}.}[\linewidth]{
        \tikzinput{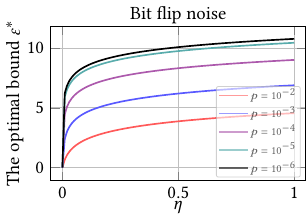}{tikz/hf_bf.tex}
    }
    
    \caption{Comparison of $\epsilon$-differential privacy on \textit{Variational Quantum Eigensolver
Algorithms} with different noise levels.}\label{fig_VQE}
\end{figure}

\begin{figure}[!ht]
    \centering
    
    \subcaptionbox{The optimal bound function $\epsilon^*(\eta)$ for \emph{QAOA\_21}.}[\linewidth]{
        \tikzinput{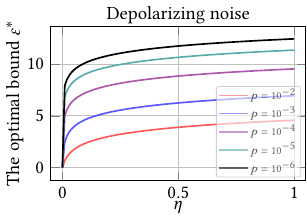}{tikz/qaoa_dp.tex}
    }
    
    \subcaptionbox{The optimal bound function $\epsilon^*(\eta)$ for \emph{QAOA\_21}.}[\linewidth]{
        \tikzinput{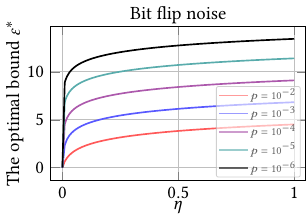}{tikz/qaoa_bf}
    }
    
    \caption{Comparison of $\epsilon$-differential privacy on \textit{Quantum Approximate Optimization
Algorithms} with different noise levels.}\label{fig_QAOA}
\end{figure}

\begin{figure}[!ht]
    \centering
    
    \subcaptionbox{The optimal bound function $\epsilon^*(\eta)$ for \emph{inst\_4x5\_10}.}[\linewidth]{
        \tikzinput{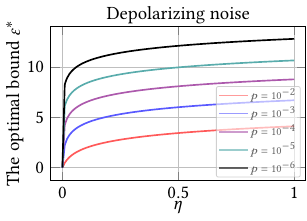}{tikz/inst_dp}
    }
    
    \subcaptionbox{The optimal bound function $\epsilon^*(\eta)$ for \emph{inst\_4x5\_10}.}[\linewidth]{
        \tikzinput{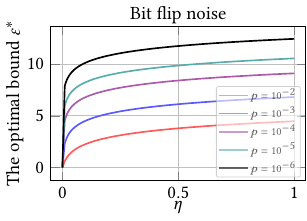}{tikz/inst_bf}
    }\hfill

    \caption{Comparison of $\epsilon$-differential privacy on \textit{Quantum Supremacy Algorithms} with different noise levels.}
\end{figure}

\begin{figure*}[!ht]
    \centering
    \pgfplotsset{width=5.9cm,height=4.2cm,legend style={opacity=0.6},y label style={yshift=-5pt}}

    \hfill
    \subcaptionbox{The optimal bound function $\epsilon^*(\eta)$ for \emph{MNIST\_10}.}[0.31\linewidth]{
        \tikzinput{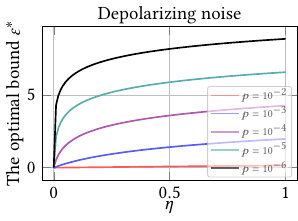}{tikz/mnist_dp}
    }
    \hfill
    \subcaptionbox{The optimal bound function $\epsilon^*(\eta)$ for \emph{GC\_9}.}[0.31\linewidth]{
        \tikzinput{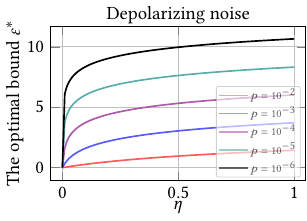}{tikz/gc_dp}
    }
    \hfill
    \subcaptionbox{The optimal bound function $\epsilon^*(\eta)$ for \emph{EC\_9}.}[0.31\linewidth]{
        \tikzinput{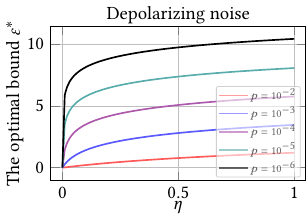}{tikz/ec_dp}
    }
    \hfill
    
    \hfill
    \subcaptionbox{ The optimal bound function $\epsilon^*(\eta)$ for  \emph{MNIST\_10}.}[0.31\linewidth]{
        \tikzinput{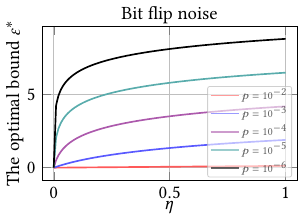}{tikz/mnist_bf}
    }
    \hfill
    \subcaptionbox{The optimal bound function $\epsilon^*(\eta)$ for \emph{GC\_9}.}[0.31\linewidth]{
        \tikzinput{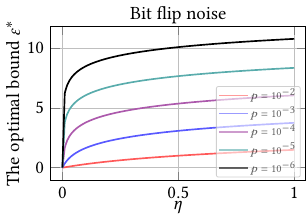}{tikz/gc_bf}
    }
    \hfill
    \subcaptionbox{The optimal bound function $\epsilon^*(\eta)$ for \emph{EC\_9}.}[0.31\linewidth]{
        \tikzinput{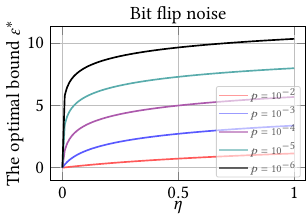}{tikz/ec_bf}
    }
    \hfill

    \caption{Comparison of $\epsilon$-differential privacy on various \textit{Quantum Machine Learning Models} with different noise levels.}\label{fig_QML}
\end{figure*}

{\vskip 3pt}

{\bf Optimal Bound Function $\epsilon^*(\eta)$:} After the above verification process, we have the values of $\kappa^*$ for all experiments. We choose the one in every kind of experiment with the largest qubits as the benchmark to depict  the optimal bound function $\epsilon^*(\eta)$ in Figs.~\ref{fig_VQE}-\ref{fig_QML}, respectively. At the same time, we add more noise levels to further explore the tendency of the optimal bound function $\epsilon^*(\eta)$. All experimental results confirm that the quantum noises on input states  can logarithmically enhance  the differential privacy as we claimed before.  Furthermore, as quantum differential privacy protects the privacy of encoded classical data, as shown in Example~\ref{Exa:angle_encoding}, introducing quantum noise can further enhance the differential privacy of the encoded data, much like how adding classical noise improves the privacy of original classical data~\cite{dwork2014algorithmic}.

\section{Conclusion}

In this paper, we established a formal framework for detecting violations of  differential privacy for quantum algorithms. In particular, we developed an algorithm to not only verify whether or not a quantum algorithm is differentially private but also provide counterexamples when the privacy is unsatisfied. The counterexample consists of a pair of quantum states violating the privacy to reveal the cause of the violation. For practicability, we implemented our algorithm on TensorFlow Quantum and TorchQuantum, the quantum extensions of famous machine learning platforms --- TensorFlow and PyTorch, respectively. Furthermore, for scalability,  we adapted Tensor Networks (a highly efficient data structure) in our algorithm to overcome the state explosion (the complexity of the algorithm is exponential with the number of qubits) such that the practical performance of our algorithm can be improved. The effectiveness and efficiency of our algorithm were tested by numerical experiments on a bulk of quantum algorithms ranging from quantum supremacy (beyond classical computation) algorithms to quantum machine learning models with up to 21 qubits, which all have been implemented on current quantum hardware devices. The experimental results showed that quantum differential privacy can benefit from adding quantum noises on either quantum circuits or input states, which is consistent with the obtained theoretical results presented as Theorem~\ref{Thm:verifier}.  

For future works, extending the techniques developed for quantum algorithms in this paper to verify the differential privacy for quantum databases is an interesting research topic for protecting the privacy of quantum databases. As we discussed in  Section~\ref{Sec:definition}, the neighboring relation for defining the differential privacy of quantum databases is the reachability between two quantum states by performing a quantum operation
(super-operator) on a single quantum bit only~\cite{aaronson2019gentle}, while that for our setting in this paper is the trace distance of two quantum states. Due to this fundamental difference in the neighboring relation,  additional extensions will be required such as developing a reachability-based search algorithm to find the violations of the differential privacy for quantum databases. Another challenging research line is to study how to train a quantum machine learning algorithm with a differential privacy guarantee. This has been done for classical machine learning algorithms~\cite{abadi2016deep}, but untouched at all for quantum algorithms. 
\section*{Acknowledgments}
This work was partly supported by the Youth Innovation Promotion Association CAS, the National Natural Science Foundation of China (Grant No. 61832015), the Young Scientists Fund of the National Natural Science Foundation of China (Grant No. 62002349), the Key Research Program of the Chinese Academy of Sciences (Grant No. ZDRW-XX-2022-1).



\bibliographystyle{unsrt} 
\balance
\bibliography{main}


\appendix
\onecolumn
\section*{Appendix}
We begin with the proof of Theorem~\ref{Thm:verifier} as it will be used to prove the others. 
\section{The Proof of Theorem~\ref{Thm:verifier}}
\begin{proof}
    For the first claim, by the definition of differential privacy in Definition~\ref{def:DP}, we have that for all quantum states $\rho,\sigma\in\dh$ with $D(\rho,\sigma)\leq \eta$ and $\s\subseteq \o$, 
    \begin{equation*}
        \begin{aligned}
        &\sum_{k\in\s}\tr(M_{k}\e(\rho))\leq \exp(\epsilon)\sum_{k\in\s}\tr(M_{k}\e(\sigma))+\delta\\
        \Leftrightarrow{} & \delta \geq \sum_{k\in\s}\tr(\e^\dagger(M_k)(\rho-e^\epsilon\sigma))\\
        \Leftrightarrow{} & \delta \geq \tr(M_\s(\rho-e^\epsilon\sigma)).
        \end{aligned}
    \end{equation*}
    By the arbitrariness of $ \s$, the above inequality holds if and only if for any   $\s\subseteq\o$, we have $\delta\geq\delta_{\s}$, where 
    \begin{equation}\label{eq:deltaS}
        \begin{aligned}
            \delta_{\s}=\sup_{\rho,\sigma\text{ with }D(\rho,\sigma)\leq \eta}\tr(M_\s(\rho-e^\epsilon\sigma)).
        \end{aligned}
    \end{equation}
    Next, we claim that 
    $$\delta_{\s}= \eta\lambda_{\max}(M_{\s})-(e^\epsilon+\eta-1)\lambda_{\min}(M_{\s}).$$
    
    First, let $\ket{\psi}$ and $\ket{\phi}$ be two normalized eigenvectors of $M_{\s}$ corresponding to the maximum and minimum eigenvalues, respectively, and 
    \[\gamma=\eta\ketbra{\psi}{\psi}+(1-\eta)\ketbra{\phi}{\phi}\qquad \phi=\ketbra{\phi}{\phi}.\]
    Then, by the arbitrariness of $\rho, \sigma$ and $D(\gamma,\phi)=\eta$, we have that 
    \begin{equation*}
        \begin{aligned}\delta_{\s}\geq&\tr(M_{\s}(\gamma-e^{\epsilon}\phi))\\
        =&\eta\tr(M_{\s}\psi)+(1-\eta-e^{\epsilon})\tr(M_{\s}\phi)\\
        =&\eta\lambda_{\max}(M_{\s})-(e^\epsilon+\eta-1)\lambda_{\min}(M_{\s}).
        \end{aligned}
    \end{equation*}
        
    On the other hand, for any quantum states $\rho,\sigma\in\dh$ with $D(\rho,\sigma)\leq \eta$, let $\rho-\sigma=\Delta_{+}-\Delta_{-}$ be a decomposition into orthogonal positive and negative parts (i.e., $\Delta_{\pm}\geq 0$ and $\Delta_{+}\Delta_{-}=0$). Then we have $\tr(\Delta_+)=\tr(\Delta_{-})$ because $0=\tr(\rho-\sigma)=\tr(\Delta_{+}-\Delta_{-})$, and then $\frac{1}{2}\tr(|\rho-\sigma|)=\tr(\Delta_{+})$. Therefore, $D(\rho,\sigma)=\tr(\Delta_{+})\leq\eta$. Furthermore, 
    \begin{equation*}
        \begin{aligned}                    \delta_{\s}=&\sup_{\rho,\sigma\text{ with }D(\rho,\sigma)\leq \eta}\tr(M_\s(\rho-e^\epsilon\sigma))\\
        =&\sup_{\tr(\Delta_{-})=\tr(\Delta_{+})\leq \eta }\tr[M_\s(\Delta_{+}-\Delta_{-}+\sigma-e^\epsilon\sigma)] \\  
        =&\sup_{\tr(\Delta_{-})=\tr(\Delta_{+})\leq \eta }\tr(M_\s\Delta_{+})-\tr(M_\s\Delta_{-})-(e^\epsilon-1)\tr(M_\s\sigma) \\ 
        =&\sup_{\tr(\Delta_{-})=\tr(\Delta_{+})\leq \eta }\tr(\Delta_{+})\tr(M_\s\frac{\Delta_{+}}{\tr(\Delta_{+})})\\
        &-\tr(\Delta_{-})\tr(M_\s\frac{\Delta_{-}}{\tr(\Delta_{-})})-(e^\epsilon-1)\tr(M_\s\sigma) \\ 
        \leq &\sup_{\tr(\Delta_{-})=\tr(\Delta_{+})\leq \eta}\tr(\Delta_{+})\max_{\rho_1\in\dh}\tr(M_{\s}\rho_1)-\tr(\Delta_{-})\min_{\rho_2\in\dh}\tr(M_{\s}\rho_2)-(e^\epsilon-1)\min_{\rho_3\in\dh}\tr(M_{\s}\rho_3)\\
        =&\sup_{\tr(\Delta_{-})=\tr(\Delta_{+})\leq \eta}\tr(\Delta_{+})[\lambda_{\max}(M_\s)-\lambda_{\min}(M_\s)]-(e^\epsilon-1)\lambda_{\min}(M_{\s})\\
        \leq &\eta\lambda_{\max}(M_{\s})-(e^\epsilon+\eta-1)\lambda_{\min}(M_{\s}).
        \end{aligned}
    \end{equation*}
    In summary, $\delta_{\s}=\eta\lambda_{\max}(M_{\s})-(e^\epsilon+\eta-1)\lambda_{\min}(M_{\s})$. With this, we can complete the first claim in the theorem. 

    For proving the second one, we only need to set $\delta=0$ in the first claim and solve the inequalities $0\geq \delta_{\s}$ for any $\s\subseteq \o$.
\end{proof}

\section{The Proof of Theorem~\ref{Thm:post-processing}}
\begin{lemma}\label{lem:inequality}
    Let $\e$  be a super-operator on Hilbert space $\h$. Then for any positive semi-definite matrix $M$, we have 
    \[\lambda_{\min}(M)\leq \lambda_{\min}(\e^\dagger(M))\leq \lambda_{\max}(\e^\dagger(M))\leq\lambda_{\max}(M).\]
\end{lemma}
\begin{proof}
    First, we have that 
    \[\lambda_{\max}(M)=\max_{\rho\in\dh}\tr(M\rho) \text{ and  } \lambda_{\min}(M)=\min_{\rho\in\dh}\tr(M\rho).\]
    Thus, 
    \begin{equation*}
        \begin{aligned}
            &\lambda_{\max}(\e^{\dagger}(M))=\max_{\rho\in\dh}\tr(\e^{\dagger}(M)\rho) \\ &\lambda_{\min}(\e^{\dagger}(M))=\min_{\rho\in\dh}\tr(\e^{\dagger}(M)\rho).
        \end{aligned}
    \end{equation*}
    Furthermore, as $\tr(\e^{\dagger}(M)\rho)=\tr(M\e(\rho))$,
    \begin{equation*}
        \begin{aligned}
            &\lambda_{\max}(\e^{\dagger}(M))=\max_{\rho\in\dh}\tr(M\e(\rho)) \\ &\lambda_{\min}(\e^{\dagger}(M))=\min_{\rho\in\dh}\tr(M\e(\rho)).
        \end{aligned}
    \end{equation*}
    Let $\s=\{\e(\rho)|\rho\in\dh\}$. Then 
    \[\lambda_{\max}(\e^{\dagger}(M))=\max_{\rho\in\s}\tr(M\rho) \text{ and  } \lambda_{\min}(\e^{\dagger}(M))=\min_{\rho\in\s}\tr(M\rho).\]
    Because of $\s\subseteq\dh$, we obtain that $\lambda_{\min}(M)\leq \lambda_{\min}(\e^\dagger(M))$ and $\lambda_{\max}(\e^\dagger(M))\leq\lambda_{\max}(M)$, completing the proof. 
\end{proof}
Now  we can prove Theorem~\ref{Thm:post-processing}.
\begin{proof}
    It follows from Theorem~\ref{Thm:verifier} and Lemma~\ref{lem:inequality} by noting that $\lambda_{\max}(\u^\dagger(M))=\lambda_{\max}(M)$ and $\lambda_{\min}(\u^\dagger(M))=\lambda_{\min}(M)$ for any positive semi-definite matrix $M$.
\end{proof}

\section{The Proof of Theorem~\ref{Thm:composition}}
\begin{proof}
First, we prove Claim 1.
For $k\in\{1,2\},$ we redefine  ${M}_{k,\s_k}=\sum_{j_k\in \s_k}\e_k^{\dagger}(M_{k,j_k})$.
   Then by Theorem~\ref{Thm:verifier}, 
    \[\max_{\s_1\subseteq \o_1}\kappa(M_{1,\s_1})\leq \frac{e^{\epsilon_1} +\eta_1 -1}{\eta_1}\qquad \max_{\s_2\subseteq \o_1}\kappa(M_{2,\s_2})\leq \frac{e^{\epsilon_2} +\eta_2 -1}{\eta_2}.\]
    With the above two inequalities, we have 
    \begin{equation*}
       \begin{aligned}
            &\max_{\s_1\subseteq \o_1}\max_{\s_2\subseteq \o_2}\kappa(M_{1,\s_1})\kappa(M_{2,\s_2})\\
           \leq&\frac{e^{\epsilon_1} +\eta_1 -1}{\eta_1}\frac{e^{\epsilon_2} +\eta_2 -1}{\eta_2}\\
           =&\frac{e^{\epsilon_1+\epsilon_2}+(\eta_1-1)(\eta_2-1)+e^{\epsilon_1}(\eta_2-1)+e^{\epsilon_2}(\eta_1-1)}{\eta_1\eta_2}\\
           \leq &\frac{e^{\epsilon_1+\epsilon_2}+(\eta_1-1)(\eta_2-1)+\eta_1+\eta_2-2}{\eta_1\eta_2}\\        =&\frac{e^{\epsilon_1+\epsilon_2}+\eta_1\eta_2-1}{\eta_1\eta_2}.
        \end{aligned}
    \end{equation*}
   The last inequality results from $e^\epsilon\geq 1$ for any $\epsilon\geq 0$ and $\eta_1, \eta_2\leq 1$. 


Next, we prove Claim 2.

For any subsets $\s_1\subseteq \o_1$ and $\s_2\subseteq \o_2$, we have
    \begin{equation*}
        \begin{aligned}
           &\eta_{1}\eta_2\lambda_{\max}(M_{1,\s_1}\otimes M_{2,\s_2})-(e^{\epsilon_1+\epsilon_2}+\eta_1\eta_2-1)\lambda_{\min}(M_{1,\s_1}\otimes M_{2,\s_2})\\
           =&\eta_{1}\eta_2\lambda_{\max}(M_{1,\s_1})\lambda_{\max} (M_{2,\s_2})-(e^{\epsilon_1+\epsilon_2}+\eta_1\eta_2-1)\lambda_{\min}(M_{1,\s_1})\lambda_{\min}( M_{2,\s_2})\\
           \leq &\eta_{1}\eta_2\lambda_{\max}(M_{1,\s_1})\lambda_{\max} (M_{2,\s_2})-(e^{\epsilon_1}+\eta_1-1)(e^{\epsilon_2}+\eta_2-1)\lambda_{\min}(M_{1,\s_1})\lambda_{\min}( M_{2,\s_2})\\
          \leq&\eta_{1}\eta_2\lambda_{\max}(M_{1,\s_1})\lambda_{\max} (M_{2,\s_2})-(\eta_1\lambda_{\max}(M_{1,\s_1})-\delta_1)(\eta_2\lambda_{\max}(M_{2,\s_2})-\delta_2)\\
          =&\eta_{1}\lambda_{\max}(M_{1,\s_1})\delta_2+\eta_{2}\lambda_{\max}(M_{2,\s_2})\delta_1-\delta_1\delta_2\\
          \leq &\delta_1+\delta_2.
        \end{aligned}
    \end{equation*}
   The first inequality results from $(e^{\epsilon_1+\epsilon_2}+\eta_1\eta_2-1)\geq (e^{\epsilon_1}+\eta_1-1)(e^{\epsilon_2}+\eta_2-1)$, the second one  follows from Theorem~\ref{Thm:verifier}, and the last one comes from the fact that $0\leq \eta_1,\eta_2,\lambda_{\max}(M_{1,\s_1}),\lambda_{\max}(M_{2,\s_2})\leq 1$.
\end{proof}

\section{The Proof of Theorem \ref{Thm:generator}}
\begin{proof}
    It follows from the proof of Theorem~\ref{Thm:verifier} that $\delta_{\s}$ defined by Eq.~(\ref{eq:deltaS}) can be reached by the following pair of quantum states:
    \[(\eta\ketbra{\psi}{\psi}+(1-\eta)\ketbra{\phi}{\phi}),\ketbra{\phi}{\phi})\]
    for any normalized eigenvectors $\ket{\psi}$ and $\ket{\phi}$ of $M_\s$ corresponding to the maximum and minimum eigenvalues, respectively.
\end{proof}

\end{document}